%% file: main.tex
\documentclass[10pt,journal,english,twocolumn]{IEEEtran}

\input{config}

\input{defs}

\title{RISnet: A Domain-Knowledge Driven Neural Network Architecture for RIS Optimization with Mutual Coupling and Partial CSI}
\author{
Bile Peng, \IEEEmembership{Senior Member,~IEEE}, 
Karl-Ludwig Besser, \IEEEmembership{Member,~IEEE}, 
Shanpu Shen, \IEEEmembership{Senior Member,~IEEE}, 
Finn Siegismund-Poschmann,
Ramprasad Raghunath, \IEEEmembership{Student Member,~IEEE}, 
Daniel Mittleman, \IEEEmembership{Fellow,~IEEE},\\
Vahid Jamali, \IEEEmembership{Senior Member,~IEEE}, and
Eduard A. Jorswieck, \IEEEmembership{Fellow,~IEEE}
\thanks{Parts of this work were presented at the 2023 IEEE Global Communications Conference (Globecom)~\cite{Peng2023risnet}.}
\thanks{%
B.~Peng, R.~Raghunath and E.~A.~Jorswieck are with Institute for Communications Technology, TU Braunschweig, Germany (e-mail: \{b.peng, r.raghunath, e.jorswieck\}@tu-braunschweig.de).
K.-L.~Besser was with the Department of Electrical and Computer Engineering, Princeton University,  Princeton, NJ 08544, USA, and is now with the Department of Electrical Engineering, Linköping University, Linköping, Sweden (email: {karl-ludwig.besser}@liu.se).
S.~Shen is with Department of Electrical Engineering and Electronics,
University of Liverpool, UK (email: Shanpu.Shen@liverpool.ac.uk).
F.~Siegismund-Poschmann is with Institute for Computer Science, Freie Universität Berlin, Germany (email: finn.siegismund-poschmann@fu-berlin.de).
V.~Jamali is with the Department of Electrical Engineering and Information Technology, TU Darmstadt, Germany (email: vahid.jamali@tu-darmstadt.de).
D.~Mittleman is with School of Engineering, Brown University, USA (email: daniel\_mittleman@brown.edu).
}
\thanks{The work of E.~Jorswieck, R.~Raghunath, and B.~Peng was supported partly by the Federal Ministry of Education and Research (BMBF), Germany, through the Program of Souverän, Digital, and Vernetzt Joint Project 6G-RIC under Grant 16KISK031.
The work of K.-L.~Besser was supported by the German Research Foundation (DFG) under grant BE\,8098/1-1.
The work of V.~Jamali is supported in part by the DFG within the Collaborative Research Center MAKI (SFB~1053, Project-ID~210487104) and in part by the LOEWE initiative (Hesse, Germany) within the emergenCITY center [LOEWE/1/12/519/03/05.001(0016)/72].
The work of D.~Mittleman is supported by the US National Science Foundation (CNS-1954780, CNS-2211616), the US Air Force Office of Scientific Research (FA9550-22–1-0412), the Alexander von Humboldt Foundation, 
and the DFG through a Mercator Fellowship.
}
}

\begin{document}

\maketitle

\input{abstract}

\input{intro}
\input{channel}
\input{problem}
\input{risnet}

\input{results}
\input{conclusion}
\appendices
\input{appendix}

\printbibliography[heading=bibintoc]

\end{document}

%% file: config.tex
\pdfoutput=1

\usepackage{babel}
\usepackage[babel]{microtype}
\usepackage[babel]{csquotes}

\usepackage[svgnames]{xcolor}
\usepackage{graphicx}

\usepackage{amsmath}
\usepackage{amsfonts}
\usepackage{amssymb}
\usepackage{amsthm}
\usepackage{algorithm,algpseudocode}
\usepackage{bm}
\usepackage{siunitx}
\usepackage{booktabs}
\usepackage{subfigure}
\usepackage{tabularx}
\usepackage{titlesec} %

\usepackage{tikz}
\usepackage{pgfplots}
\pgfplotsset{compat=newest}
\usetikzlibrary{positioning}
\usetikzlibrary{shapes.geometric, arrows, decorations.pathreplacing}
\usetikzlibrary{shapes.arrows,arrows.meta}
\usetikzlibrary{3d}
\usetikzlibrary{calc}

\usepackage[defernumbers=false,
backend=biber,
isbn=false,
url=false,
maxbibnames=999,
maxcitenames=10,
doi=false,
giveninits=true,
sortcites=true,
eprint=true,
hyperref=true,
style=ieee,
]{biblatex}
\bibliography{abbr-ref.bib}
\renewcommand*{\bibfont}{\footnotesize}

\usepackage{hyperref}
\hypersetup{%
  pdftitle = {RISnet: A Domain-Knowledge Driven Neural Network Architecture for RIS Optimization with Mutual Coupling and Partial CSI},
  pdfauthor = {Bile Peng, Karl-Ludwig Besser, Shanpu Shen, Finn Siegismund-Poschmann, Ramprasad Raghunath, Daniel Mittleman, Vahid Jamali, Eduard A. Jorswieck},
  pdflang = {en-US},
  colorlinks,
  allcolors=.,
  urlcolor=MidnightBlue,
  bookmarksnumbered,
  bookmarksopen,
}

\usepackage[acronyms,nomain,xindy,nowarn]{glossaries}
\makeglossaries
\loadglsentries{acronyms.tex}
\setacronymstyle{long-short}
\glsdisablehyper

\addto\extrasenglish{%

}
\addto\extrasenglish{%

}

\newtheorem{theorem}{Theorem}

\theoremstyle{remark}
\newtheorem{remark}{Remark}

\theoremstyle{definition}
\newtheorem{definition}{Definition}

\newenvironment{change}{}{\ignorespacesafterend}
\newenvironment{changev2}{}{\ignorespacesafterend}

\IfPackageLoadedTF{titlesec}{%
	\ifCLASSOPTIONconference%
	\titlespacing{\section}{0pt}{1.5ex plus 1.5ex minus 0.5ex}{0.7ex plus 1ex minus 0ex} %
	\titlespacing{\subsection}{0pt}{1.5ex plus 1.5ex minus 0.5ex}{0.7ex plus .5ex minus 0ex} %
	\else
	\titlespacing{\section}{0pt}{3.0ex plus 1.5ex minus 1.5ex}{0.7ex plus 1ex minus 0ex} %
	\titlespacing{\subsection}{0pt}{3.5ex plus 1.5ex minus 1.5ex}{0.7ex plus .5ex minus 0ex} %
	\fi
	
	\def\thesubsubsectiondis{\arabic{subsubsection})}
	\def\theparagraphdis{\alph{paragraph})}
	\titleformat{\subsubsection}[runin]{\normalfont\normalsize\itshape}{\thesubsubsectiondis}{.5em}{}[:]
	\titlespacing*{\subsubsection}{\parindent}{0ex plus 0.1ex minus 0.1ex}{1ex}
	\titleformat{\paragraph}[runin]{\normalfont\normalsize\itshape}{\theparagraphdis}{.5em}{}[:]
	\titlespacing*{\paragraph}{2\parindent}{0ex plus 0.1ex minus 0.1ex}{1ex}
}{}

%% file: acronyms.tex
\newacronym{isac}{ISAC}{integrated sensing and communication}
\newacronym{ae}{AE}{auto-encoder}
\newacronym{vae}{VAE}{variational auto-encoder}
\newacronym{ap}{AP}{access point}
\newacronym{aoa}{AoA}{angle-of-arrival}
\newacronym{ao}{AO}{alternating optimization}
\newacronym{bs}{BS}{base station}
\newacronym{dnn}{DNN}{deep neural network}
\newacronym{cnn}{CNN}{convolutional neural network}
\newacronym{rnn}{RNN}{recurrent neural network}
\newacronym{ma}{MA}{multiple access}
\newacronym{nn}{NN}{neural network}
\newacronym{ue}{UE}{user equipment}
\newacronym{noma}{NOMA}{non-orthogonal multiple access}
\newacronym{mimo}{MIMO}{multiple-input-multiple-output}
\newacronym{miso}{MISO}{multiple-input-single-output}
\newacronym{cf}{CF mMIMO}{cell-free massive MIMO}
\newacronym{cdf}{CDF}{cumulative distribution function}
\newacronym{gnn}{GNN}{graph neural network}
\newacronym{cc}{CC}{channel charting}
\newacronym{los}{LoS}{line-of-sight}
\newacronym{csi}{CSI}{channel state information}
\newacronym{star}{STAR}{simultaneously transmitting and reflecting}
\newacronym{aod}{AoD}{angle-of-departure}
\newacronym{lwf}{LwF}{learing without forgetting}
\newacronym{sinr}{SINR}{signal-to-interference-plus-noise ratio}
\newacronym{snr}{SNR}{signal-to-noise ratio}
\newacronym{sgd}{SGD}{stochastic gradient descent}
\newacronym{agi}{AGI}{artificial general intelligence}
\newacronym{llm}{LLM}{large language model}
\newacronym{dl}{DL}{deep learning}
\newacronym{rl}{RL}{reinforcement learning}
\newacronym{crb}{CRB}{Cram\'er-Rao bound}
\newacronym{tpu}{TPU}{tensor processing unit}
\newacronym{evt}{EVT}{extreme value theory}
\newacronym{urllc}{URLLC}{ultra-reliable low latency communications}
\newacronym{mrt}{MRT}{maximum ratio transmission}
\newacronym{embb}{eMBB}{enhanced mobile broadband}
\newacronym{mmse}{MMSE}{minimum mean square error}
\newacronym{qd}{QD}{quasi degradation}
\newacronym{rss}{RSS}{received signal strength}
\newacronym{siso}{SISO}{single-input single-output}
\newacronym{rsma}{RSMA}{rate-splitting multiple access}
\newacronym{lstm}{LSTM}{long short-term memory}
\newacronym{drl}{DRL}{deep reinforcement learning}
\newacronym{dqn}{DQN}{deep Q-network}
\newacronym{ml}{ML}{machine learning}
\newacronym{sca}{SCA}{successive convex approximation}
\newacronym{admm}{ADMM}{alternating direction method of multipliers}
\newacronym{mm}{MM}{majorization-minimization}
\newacronym{rmcg}{RMCG}{Riemannian manifold conjugate gradient}
\newacronym{sdma}{SDMA}{space-division multiple access}
\newacronym{sdr}{SDR}{semidefinite relaxation}
\newacronym{mse}{MSE}{mean squared error}
\newacronym{ris}{RIS}{reconfigurable intelligent surface}
\newacronym{tof}{ToF}{time of flight}
\newacronym{em}{EM}{expectation-maximization}
\newacronym{zf}{ZF}{zero-forcing}
\newacronym{ekf}{EKF}{extended Kalman filter}
\newacronym{bcd}{BCD}{block coordinate descent}
\newacronym{concrete}{CONCRETE}{CONtinuous relaxation for disCRETE}
\newacronym{mop}{MOP}{multi-objective programming}
\newacronym{uwb}{UWB}{ultra-wide band}
\newacronym{sic}{SIC}{successive interference cancellation}
\newacronym{otfs}{OTFS}{orthogonal time frequency space}
\newacronym{ofdm}{OFDM}{orthogonal frequency-division multiplexing}
\newacronym{mpc}{MPC}{multi-path component}
\newacronym{wmmse}{WMMSE}{weighted minimum mean square error}
\newacronym{wsr}{WSR}{weighted sum-rate}
\newacronym{mab}{MAB}{multi-armed bandit}
\newacronym{ppo}{PPO}{proximal policy optimization}
\newacronym{sac}{SAC}{soft actor-critic}
\newacronym{iid}{i.i.d.}{independent and identically distributed}
\newacronym{wrt}{w.r.t.}{with respect to}
\newacronym{fdma}{FDMA}{frequency domain multiple access}
\newacronym{tdma}{TDMA}{time domain multiple access}
\newacronym{cdma}{CDMA}{code domain multiple access}
\newacronym{swipt}{SWIPT}{simultaneous wireless information and power transfer}

%% file: defs.tex
\newcommand{\e}{\mathrm{e}}
\newcommand{\jj}{\mathrm{j}}
\let\imag\jj
\DeclareMathOperator{\trace}{Tr}

\newcommand{\todo}[2][]{\ignorespaces
	\if\relax\detokenize{#1}\relax
	{\color{red}[TODO: #2]}%
	\else
	{\color{red}[TODO (#1): #2]}%
	\fi
}

\definecolor{plot0}{HTML}{004488}
\definecolor{plot1}{HTML}{DDAA33}
\definecolor{plot2}{HTML}{BB5566}
\definecolor{plot3}{HTML}{000000}
\definecolor{plot4}{HTML}{AAAAAA}

%% file: abstract.tex
\begin{abstract}\noindent\boldmath
\Gls{sdma} plays an important role in modern wireless communications.
Its performance depends on the channel properties,
which can be improved by \glspl{ris}.
In this work,
we jointly optimize \gls{sdma} precoding at the \gls{bs} and \gls{ris} configuration. 
We tackle difficulties of
mutual coupling between \gls{ris} elements,
scalability to more than 1000 \gls{ris} elements,
and high requirement for channel estimation.
We first derive an \gls{ris}-assisted channel model considering mutual coupling,
then propose an unsupervised \gls{ml} approach to optimize the \gls{ris} with
a dedicated \gls{nn} architecture \emph{RISnet},
which has good scalability, desired permutation-invariance,
and a low requirement for channel estimation.
Moreover,
we leverage existing high-performance analytical precoding scheme
to propose a hybrid solution of
\gls{ml}-enabled \gls{ris} configuration and analytical precoding at \gls{bs}.
More generally,
this work is an early contribution to combine \gls{ml} technique and domain knowledge in communication for \gls{nn} architecture design.
Compared to generic \gls{ml},
the problem-specific \gls{ml} can achieve higher performance, lower complexity and permutation-invariance.
\end{abstract}

\glsresetall

%% file: intro.tex
\section{Introduction}
\label{sec:intro}

The \gls{sdma} technique plays an important role in modern multi-user wireless communications.
Its performance depends heavily on the channel condition.
For example,
a high channel gain realizes a high \gls{snr},
a high-rank \gls{mimo} channel matrix makes it possible to serve multiple users~\cite{wei2008iterative}.
In the past,
the wireless channel has been considered as given.
However, in recent years,
the \gls{ris}~\cite{liu2021reconfigurable} has been proposed to manipulate the channel property.
In this paper,
we consider the problem of joint optimization of precoding at the \gls{bs} and configuration of the \gls{ris}.

In the literature,
multiple precoding techniques have been proposed for \gls{sdma} at the \gls{bs},
including \gls{mrt}, \gls{zf}, \gls{mmse} precoding~\cite{joham2005linear},
and \gls{wmmse} precoding with proved equivalence to \gls{wsr} maximization~\cite{shi2011iteratively}.
The joint optimization of precoding and \gls{ris} configuration was performed with
 \gls{bcd}~\cite{guo2020weighted},
\gls{mm}~\cite{huang2018achievable,zhou2020intelligent} and
\gls{admm}~\cite{liu2021two} algorithms to maximize the \gls{wsr} in \gls{sdma}.
In addition,
\gls{rmcg} and Lagrangian method were applied to optimize multiple \glspl{ris} and \glspl{bs} to serve cell-edge users~\cite{li2020weighted}.
Successive refinement algorithm and exhaustive search were applied for passive beamforming improvement~\cite{wu2019beamforming}.
The active \glspl{ris} was optimized with the \gls{sca} algorithm to maximize the \gls{snr}~\cite{long2021active}.
The gradient-based optimization was applied to optimize the effective rank and the minimum singular value~\cite{Elmossallamy2021spatial}.
\begin{changev2}
A novel \gls{ao} scheme is proposed to minimize the transmit power subject to the data rate requirement~\cite{Lin2022}.
A multi-layer refracting \gls{ris} is proposed for \gls{swipt} to overcome severe large-scale fading~\cite{An2024}.
\end{changev2}

In general,
the above analytical iterative methods do not scale well with the number of \gls{ris} elements.
No more than 100 elements were assumed in \cite{guo2020weighted,huang2018achievable,liu2021two,li2020weighted,wu2019beamforming,zhou2020intelligent,long2021active}
and up to 400~\gls{ris} elements were assumed in \cite{Elmossallamy2021spatial},
which is far from the vision of more than 1000~\gls{ris} elements~\cite{di2020smart}
and the requirement in many scenarios to realize a necessary link budget~\cite{najafi2020physics}.
Another common important limitation of the analytical optimization approaches is that
suboptimal approximations were applied to make the problem 
solvable~\cite{guo2020weighted,wu2019beamforming,liu2021two}.
Moreover, the required numbers of iterations make the proposed iterative algorithms difficult to be implemented in real time
since the computation time is longer than the channel coherence time.

A noticeable effort is to apply \gls{ml} to optimize the \gls{ris},
which bypasses the difficulty of analytical solution via the universal approximation property of the \gls{nn}~\cite{hornik1989multilayer}.
Recently, \gls{dl} and \gls{rl} were applied and compared for \gls{ris} optimization~\cite{zhong2021ai}.
\Gls{lstm} and \gls{dqn} were applied to optimize \gls{ris} for \gls{noma}~\cite{gao2021machine}.
\Gls{rl} was applied to maximize the sum-rate in
\gls{sdma}~\cite{huang2020reconfigurable}, %
and
\gls{noma}~\cite{shehab2022deep,huang2020reconfigurable} and energy efficiency in \gls{noma}~\cite{guo2022energy}.
The achievable rate was predicted and the \gls{ris} was configured with \gls{dl}~\cite{sheen2021deep}.
The \gls{ris} was configured directly with received pilot signals~\cite{jiang2021learning}.
The mapping from received pilot signal to the phase shifts was optimized~\cite{ozdougan2020deep}.
Due to the separation of training and testing phases,
the trained \gls{ml} model was able to be applied in real time.
However, the scalability with the number of \gls{ris} elements was still limited~\cite{gao2021machine,shehab2022deep,yang2021reconfigurable,huang2020reconfigurable,jiang2021learning,ozdougan2020deep}.

The second common limitation of many works on \gls{ris} is the full \gls{csi} assumption (e.g., \cite{guo2020weighted,zhou2020intelligent,liu2021two,wu2019beamforming,long2021active}).
Due to the large number of elements, the full \gls{csi} of all elements is very difficult to obtain in real time.
Possible countermeasures are, e.g., codebook-based \gls{ris} optimization~\cite{najafi2020physics,an2022codebook}.
However, the beam training is still a major limitation.

The third common limitation of many works 
is the assumption of perfect \gls{ris} without mutual coupling.
Due to the small distance between two adjacent elements,
there might exist certain mutual coupling in the \gls{ris}.
In the literature,
a \gls{siso} \gls{ris}-assisted system was optimized~\cite{qian2021mutual}.
A mutual impedance-based communication model considering mutual coupling was presented~\cite{gradoni2021end}.
A mutual coupling aware characterization and performance analysis of the \gls{ris} was introduced~\cite{pettanice2023mutual}.
The \gls{ris} architecture was modeled and characterized using scattering parameter network analysis~\cite{shen2021modeling}.
A closed-form expression of \gls{ris}-assisted \gls{mimo} channel model in a general sense,
i.e., 
considering \gls{ris} mutual coupling,
with the direct channel from \gls{bs} to users and without the unilateral approximation~\cite{nerini2023universal},
is still an open problem.

Summarizing the state-of-the-art,
we identify three limitations of current \gls{ris} research: the insufficient consideration of mutual coupling,
the poor scalability to consider more than 1000 \gls{ris} elements
and the unrealistic assumption of full \gls{csi}.

In this work, 
we propose a dedicated \gls{nn} architecture \emph{RISnet} and an unsupervised \gls{ml} approach to address these limitations.
Our contribution is four-fold as follows.
\begin{itemize}
\item We derive an \gls{ris} channel model considering the mutual coupling between \gls{ris} elements.
The derivation is based on the scattering parameter network analysis~\cite{shen2021modeling}.
A closed-form expression is derived without unilateral approximation.
\item We propose an \gls{nn} architecture \emph{RISnet} for \gls{ris} configuration.
\begin{change}
The number of RISnet parameters is independent of the number of \gls{ris} elements,
enabling a high scalability
such that we can configure 1296~\gls{ris} elements within a few milliseconds.
Compared to it,
most conventional approaches in the literature
assume no more than 100~\gls{ris} elements,
as explained above.
\end{change}
Furthermore, the RISnet is \emph{permutation-invariant},
i.e., any permutation of users in the input has no impact on the \gls{ris} phase shifts
because permutation of users has no impact on the \gls{sdma} problem.
\item 
In addition to the scalability,
the \gls{csi} is extremely difficult to obtain due to the large number of \gls{ris} elements.
We propose an improved RISnet,
where only a few \gls{ris} elements (in our paper, 16 out of 1296) are equipped with RF chains and can estimate the channel with the pilot signals from users.
We demonstrate that RISnet can configure the phase shifts of all \gls{ris} elements with the partial \gls{csi} of a few \gls{ris} elements
if the channel is sparse,
which holds mostly in reality.
In this way, a good compromise between hardware complexity and performance is achieved.
\item 
We combine \gls{ml}-enabled \gls{ris} configuration and analytical precoding.
In this way, the performance of \gls{bs} precoding is guaranteed
because of the proven performance of analytical precodings,
the difficulty of training is reduced
because we do not need to optimize precoding.
\end{itemize}
We combine domain knowledge in communication and \gls{ml} for \gls{nn} architecture and training process design.
It not only solves the \gls{ris} configuration problem,
but can also inspire solutions to other problems.
According to~\cite{dash2022review},
incorporating domain knowledge is considered as one of the three grand challenges in \gls{ml}.
We show that an \gls{nn} architecture tailored for the considered problem
is superior in scalability, complexity, performance and robustness.

\begin{change}
This work is based on our preliminary result~\cite{Peng2023risnet}.
Compared to it,
we consider the mutual coupling between the \gls{ris} elements in this work.
The mutual coupling is unavoidable due to the short distance between \gls{ris} elements.
However,
it has not been considered in most literature
because the resulting optimization problem is too difficult.
In addition,
we maximize the \gls{wsr} in this work instead of sum-rate in \cite{Peng2023risnet},
where the weights are input of the RISnet.
In this way,
we significantly improve the flexibility to choose an operational point according to the requirement.
Furthermore,
we further propose a parallel implementation scheme with tensor operations,
which allows for very efficient training and inference (application)
in a GPU.
The permutation-invariance is strictly proved.
We also discuss possibilities to combine analytical precoding techniques with \gls{ml},
depending on the differentiability of the analytical method.
\end{change}

\emph{Notations:}
$(\cdot)^+$ denotes the pseudo-inverse operation,
$|a|$ and $\arg(a)$ are amplitude and phase of complex number $a$, respectively,
$\mathbf{T}[\cdot, j, k]$ denotes the vector of the elements with position $(j, k)$ in the second and third dimensions in the three-dimensional tensor $\mathbf{T}$,
$\mathbf{1}^{a\times a}$ is the matrix of all ones with shape $a\times a$,
and $\mathbf{I}^{a \times a}$ is the identity matrix with shape $a\times a$,
i.e., all diagonal elements of $\mathbf{I}^{a \times a}$ are ones,
all off-diagonal elements are zeros.
We also define $\mathbf{E}^{a \times a} = \mathbf{1}^{a\times a} - \mathbf{I}^{a \times a}$.

%% file: channel.tex
\section{The RIS Channel Model Considering Mutual Coupling}
\label{sec:channel}

We consider a \gls{ris}-aided multi-user \gls{miso} scenario,
as depicted in \autoref{fig:system_model}.
The channel from \gls{bs} to \gls{ris} is denoted as $\mathbf{H} \in \mathbb{C}^{N\times M}$,
where $N$ is the number of \gls{ris} elements
and $M$ is the number of \gls{bs} antennas.
The channel from \gls{ris} to users is denoted as $\mathbf{G} \in \mathbb{C}^{U\times N}$,
where $U$ is the number of users.
The direct channel from \gls{bs} to users directly is denoted as $\mathbf{D} \in \mathbb{C}^{U\times M}$.

\begin{figure}[htbp]
    \centering
    \resizebox{.5\linewidth}{!}{
    \input{figs/system_model_small.tex}}
    \caption{The system model of \gls{ris}-assisted downlink multi-user broadcasting.}
    \label{fig:system_model}
\end{figure}
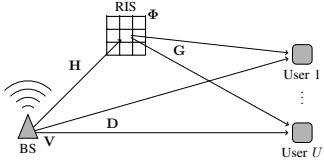

Denote the precoding matrix as $\mathbf{V} \in \mathbb{C}^{M\times U}$
and the signal processing matrix of the \gls{ris} as the diagonal matrix~${\boldsymbol{\Phi} \in \mathbb{C}^{N\times N}}$, 
where the diagonal element in row~$n$ and column~$n$ is $\phi_{nn}=\e^{\imag\psi_n}$, 
with $\psi_n \in [0, 2\pi)$ being the phase shift of \gls{ris} element~$n$. 
Conventionally,
the channel~$\mathbf{C}$ between \gls{bs} and users is~\cite{guo2020weighted}
\begin{equation}
    \mathbf{C} = \mathbf{D} + \mathbf{G} \boldsymbol{\Phi}\mathbf{H},
    \label{eq:reduced_channel}
\end{equation}
and the signal received at the users is
\begin{equation}
\mathbf{y} =  \mathbf{C}\mathbf{V} \mathbf{x} + \mathbf{n},
\label{eq:transmission_los}
\end{equation}
where $\mathbf{x} \in \mathbb{C}^{U \times 1}$ is the transmitted symbols,
$\mathbf{y} \in \mathbb{C}^{U \times 1}$ is the received symbols, and $\mathbf{n} \in \mathbb{C}^{U \times 1}$ is the noise. 

An implicit assumption of \eqref{eq:reduced_channel} is that the \gls{ris} elements do not have mutual coupling.
However,
this assumption might not hold
due to the small distance between \gls{ris} elements.
In the following,
we assume transmitter and receiver without mutual coupling and an \gls{ris} with mutual coupling.

Our derivation is a generalization of Section~III in~\cite{shen2021modeling}.
We define the S-parameter matrix~$\mathbf{S}$ for the signal transmission system shown in \autoref{fig:system_model} and partition it as
\begin{equation}
    \mathbf{S} = \left[
    \begin{array}{ccc}
        \mathbf{S}_{TT} & \mathbf{S}_{TI} & \mathbf{S}_{TR} \\
        \mathbf{S}_{IT} & \mathbf{S}_{II} & \mathbf{S}_{IR} \\
        \mathbf{S}_{RT} & \mathbf{S}_{RI} & \mathbf{S}_{RR} \\
    \end{array}
    \right],
      \label{eq:s_partition}
\end{equation}
where index~$T$ stands for transmitter,
$I$ stands for \gls{ris}
and $R$ stands for receiver.
The diagonal blocks~$\mathbf{S}_{TT}$, $\mathbf{S}_{II}$ and $\mathbf{S}_{RR}$ are the S-matrices of transmitter, \gls{ris} and receiver, respectively.
The off-diagonal blocks~$\mathbf{S}_{TI}$, $\mathbf{S}_{TR}$ and $\mathbf{S}_{IR}$ are the channels between transmitter and \gls{ris}, transmitter and receivers, and \gls{ris} and receivers, respectively,
i.e., the channels~$\mathbf{H}$, $\mathbf{D}$ and $\mathbf{G}$ in \autoref{fig:system_model}, respectively.
We also have the diagonal reflection coefficient matrix~$\boldsymbol{\Lambda}$,
which is defined as
\begin{equation}
    \boldsymbol{\Lambda}=\left[
    \begin{array}{ccc}
    \boldsymbol{\Lambda}_T & \mathbf{0} & \mathbf{0}\\
    \mathbf{0} & \boldsymbol{\Phi} & \mathbf{0}\\
    \mathbf{0} & \mathbf{0} & \boldsymbol{\Lambda}_R\\
    \end{array}
    \right],
    \label{eq:gamma_def}
\end{equation}
where $\boldsymbol{\Lambda}_T$, $\boldsymbol{\Phi}$ and $\boldsymbol{\Lambda}_R$ are the reflection coefficient matrices of transmitter, \gls{ris} and receiver, respectively.
Define $\mathbf{T}=\mathbf{S}(\mathbf{I} - \boldsymbol{\Lambda}\mathbf{S})^{-1}$ and partition $\mathbf{T}$ as
\begin{equation}
    \mathbf{T} = \left[
    \begin{array}{ccc}
        \mathbf{T}_{TT} & \mathbf{T}_{TI} & \mathbf{T}_{TR} \\
        \mathbf{T}_{IT} & \mathbf{T}_{II} & \mathbf{T}_{IR} \\
        \mathbf{T}_{RT} & \mathbf{T}_{RI} & \mathbf{T}_{RR} \\
    \end{array}
    \right]
    \label{eq:t_partition}
\end{equation}
in the same way as~\eqref{eq:s_partition}.
According to~\cite{shen2021modeling},
the channel matrix is given by
\begin{equation}
    \mathbf{C}=(\boldsymbol{\Lambda}_R + \mathbf{I})\mathbf{T}_{RT}(\mathbf{I} + \boldsymbol{\Lambda}_T\mathbf{T}_{TT} + \mathbf{T}_{TT})^{-1}.
    \label{eq:channel_general}
\end{equation}
Although~\eqref{eq:channel_general} is the most general form of the channel,
it is too complicated for the optimization due to the matrix multiplication and inversion.
In the following, we simplify~\eqref{eq:channel_general} 
assuming no mutual coupling at transmitter and receiver,
i.e., ${\mathbf{S}_{TT}=\mathbf{0}}$ and ${\mathbf{S}_{RR}=\mathbf{0}}$,
and mutual coupling at \gls{ris}, 
i.e., ${\mathbf{S}_{II}\neq \mathbf{0}}$,
the S-parameter matrix is
\begin{equation}
    \mathbf{S}=\left[
    \begin{array}{ccc}
        \mathbf{0} & \mathbf{S}_{TI} & \mathbf{S}_{TR} \\
        \mathbf{S}_{IT} & \mathbf{S}_{II} & \mathbf{S}_{IR} \\
        \mathbf{S}_{RT} & \mathbf{S}_{RI} & \mathbf{0} \\
    \end{array}
    \right].
    \label{eq:s}
\end{equation}

We further assume transmitter and receiver with perfect impedance matching, 
i.e., $\boldsymbol{\Lambda}_T=\mathbf{0}$ and $\boldsymbol{\Lambda}_R=\mathbf{0}$,
we have
\begin{equation}
    \boldsymbol{\Lambda}=\left[
    \begin{array}{ccc}
    \mathbf{0} & \mathbf{0} & \mathbf{0}\\
    \mathbf{0} & \boldsymbol{\Phi} & \mathbf{0}\\
    \mathbf{0} & \mathbf{0} & \mathbf{0}\\
    \end{array}
    \right].
    \label{eq:gamma}
\end{equation}

Combining~\eqref{eq:s} and~\eqref{eq:gamma}, we have
\begin{equation}
    \boldsymbol{\Lambda}\mathbf{S}=\left[
    \begin{array}{ccc}
        \mathbf{0} & \mathbf{0} & \mathbf{0} \\
        \boldsymbol{\Phi}\mathbf{S}_{IT} & \boldsymbol{\Phi}\mathbf{S}_{II} & \boldsymbol{\Phi}\mathbf{S}_{IR} \\
        \mathbf{0} & \mathbf{0} & \mathbf{0} \\
    \end{array}
    \right].
    \label{eq:gamma_times_s}
\end{equation}
Applying the Neumann series, we have
\begin{equation}
    (\mathbf{I} - \boldsymbol{\Lambda}\mathbf{S})^{-1}=\sum_{k=0}^\infty (\boldsymbol{\Lambda}\mathbf{S})^k.
    \label{eq:neumann}
\end{equation}
Combining \eqref{eq:gamma_times_s} and \eqref{eq:neumann}, we obtain the first line of \eqref{eq:sum_gamma_times_s_^k} at the top of the next page.
Applying the Neumann series again in the opposite direction, we obtain the second line of \eqref{eq:sum_gamma_times_s_^k}.
\begin{figure*}[ht]
\begin{equation}
\begin{aligned}
    (\mathbf{I} - \boldsymbol{\Lambda}\mathbf{S})^{-1}=\sum_{k=0}^\infty(\boldsymbol{\Lambda}\mathbf{S})^k&=\left[
    \begin{array}{ccc}
        \mathbf{I} & \mathbf{0} & \mathbf{0} \\
        \sum_{k=1}^\infty \left((\boldsymbol{\Phi}\mathbf{S}_{II})^{k-1}\right)\boldsymbol{\Phi}\mathbf{S}_{IT} &
        \sum_{k=0}^\infty(\boldsymbol{\Phi}\mathbf{S}_{II})^{k} & 
        \sum_{k=1}^\infty\left((\boldsymbol{\Phi}\mathbf{S}_{II})^{k-1}\right)\boldsymbol{\Phi}\mathbf{S}_{IR} \\
        \mathbf{0} & \mathbf{0} & \mathbf{I} \\
    \end{array}
    \right]\\
    &=\left[
    \begin{array}{ccc}
        \mathbf{I} & \mathbf{0} & \mathbf{0} \\
        (\mathbf{I} - \boldsymbol{\Phi}\mathbf{S}_{II})^{-1}\boldsymbol{\Phi}\mathbf{S}_{IT} & (\mathbf{I} - \boldsymbol{\Phi}\mathbf{S}_{II})^{-1} & (\mathbf{I} - \boldsymbol{\Phi}\mathbf{S}_{II})^{-1}\boldsymbol{\Phi}\mathbf{S}_{IR} \\
        \mathbf{0} & \mathbf{0} & \mathbf{I} \\
    \end{array}
    \right]
    \label{eq:sum_gamma_times_s_^k}
\end{aligned}
\end{equation}
\hrulefill
\end{figure*}
Combine~\eqref{eq:s}, \eqref{eq:t_partition}, and~\eqref{eq:sum_gamma_times_s_^k}, we have
\begin{equation}
    \mathbf{T}_{TT} = \mathbf{S}_{TI} (\mathbf{I} - \boldsymbol{\Phi}\mathbf{S}_{II})^{-1}\boldsymbol{\Phi}\mathbf{S}_{IT}
    \label{eq:TTT}
\end{equation}
and
\begin{equation}
    \mathbf{T}_{RT} = \mathbf{S}_{RT} + \mathbf{S}_{RI} (\mathbf{I} - \boldsymbol{\Phi}\mathbf{S}_{II})^{-1}\boldsymbol{\Phi}\mathbf{S}_{IT}.
\end{equation}
According to~\eqref{eq:channel_general}, the channel matrix is
\begin{multline}
    \mathbf{C} = (\mathbf{S}_{RT} + \mathbf{S}_{RI} (\mathbf{I} - \boldsymbol{\Phi}\mathbf{S}_{II})^{-1}\boldsymbol{\Phi}\mathbf{S}_{IT})\, \cdot \\
    (\mathbf{I} + \mathbf{S}_{TI} (\mathbf{I} - \boldsymbol{\Phi}\mathbf{S}_{II})^{-1}\boldsymbol{\Phi}\mathbf{S}_{IT})^{-1}
\end{multline}
where the term ${\mathbf{S}_{TI} (\mathbf{I} - \boldsymbol{\Phi}\mathbf{S}_{II})^{-1}\boldsymbol{\Phi}\mathbf{S}_{IT}}$ stands for the second order reflection from transmitter to \gls{ris} and back to transmitter,
which is negligible compared to $\mathbf{I}$ in most communication systems\footnote{
According to the physics of electromagnetic wave propagation,
i.e., the Friis transmission equation, 
only a small portion of the transmitted energy reaches the receiver 
due to the significant free-space path loss over distance. 
For instance, at a frequency of \SI{3.5}{\GHz} and a propagation distance of \SI{20}{\m}, 
the channel gain is approximately $(\lambda / (4\pi d))^2 = 1.2\times 10^{-7}$,
where $\lambda$ is the wavelength and
$d$ is the propagation distance.
The square of this gain, representing second-order reflections or higher-order effects, 
is approximately $1.44\times 10^{-14}$, 
diminishing to negligible levels. 
Consequently, the term $\mathbf{S}_{TI} (\mathbf{I} - \boldsymbol{\Phi}\mathbf{S}_{II})^{-1}\boldsymbol{\Phi}\mathbf{S}_{IT}$ 
is orders of magnitude smaller than $\mathbf{I}$ and can be ignored.
}.
If we ignore this term, we have
\begin{equation}
    \mathbf{C} = \mathbf{S}_{RT} + \mathbf{S}_{RI} (\mathbf{I} - \boldsymbol{\Phi}\mathbf{S}_{II})^{-1}\boldsymbol{\Phi}\mathbf{S}_{IT}.
    \label{eq:channel}
\end{equation}
Replacing $\mathbf{S}_{IT}$, $\mathbf{S}_{RT}$ and $\mathbf{S}_{RI}$ with more conventional $\mathbf{H}$, $\mathbf{D}$ and $\mathbf{G}$,
respectively,
in the channel model context,
we have
\begin{equation}
    \mathbf{C} = \mathbf{D} + \mathbf{G} (\mathbf{I} - \boldsymbol{\Phi}\mathbf{S}_{II})^{-1}\boldsymbol{\Phi}\mathbf{H}.
    \label{eq:channel_mutual_coupling}
\end{equation}
In particular, if the \gls{ris} does not have mutual coupling,
we have $\mathbf{S}_{II} = \mathbf{0}$ and~\eqref{eq:channel_mutual_coupling} is reduced to~\eqref{eq:reduced_channel}.
In this work, we apply the model from~\cite{gradoni2021end} to obtain the $\mathbf{S}_{II}$~matrix.

\begin{remark}
Our channel model~\eqref{eq:channel_mutual_coupling} is the same as the channel model~{(80)} in~\cite{nerini2023universal}.
However,
the derivation in~\cite{nerini2023universal} requires the unilateral assumption,
i.e.,
the channel from receiver to transmitter is ignored,
which contradicts the channel reciprocity,
i.e., channel gain from transmitter to receiver is equal to
channel gain from receiver to transmitter.
The reciprocity is an essential property of wireless channels.
We obtain the same channel model without the unilateral assumption,
making our derivation scientifically more rigorous.
\end{remark}

\begin{remark}
The matrix inverse in~\eqref{eq:channel_mutual_coupling} has a prohibitively high complexity.
Therefore,
we define $\mathbf{X}=\mathbf{G}(\mathbf{I}-\boldsymbol{\Phi}\mathbf{S}_{II})^{-1}$ and obtain $\mathbf{X}$ by solving the linear equation system
${\mathbf{X}(\mathbf{I}-\boldsymbol{\Phi}\mathbf{S}_{II})=\mathbf{G}}$.
The differentiable LU decomposition is applied to solve the equation system,
which has a significantly lower computation complexity.
The channel model~\eqref{eq:channel_mutual_coupling} becomes $\mathbf{C} = \mathbf{D} + \mathbf{X}\boldsymbol{\Phi}\mathbf{H}$.
\end{remark}

%% file: figs/system_model_small.tex
\begin{tikzpicture}
    \tikzstyle{base}=[isosceles triangle, draw, rotate=90, fill=gray!60, minimum size =.5cm]
	\tikzstyle{user}=[rectangle, draw, fill=gray!60, minimum size =.5cm, rounded corners=0.1cm]
	\tikzstyle{element}=[rectangle, fill=gray!30]
	
	\node[base,label={left:BS}] (BS) at (-3,0){};
	\draw[decoration={expanding waves,segment length=6},decorate] (BS) -- (-3,1.5);
	\node[user,label={below:User 1}] (UE1) at (4,2){};
	\node[user,label={below:User $U$}] (UE2) at (4,0){};
	\draw[step=0.33cm,thick] (-1,1.98) grid (0, 3);
	\node[label={[label distance=10]above:RIS}] (RIS) at (-0.5, 2.5) {};
	
	\node at ($(UE1)!.5!(UE2)$) {\vdots};
	\node[above right=.15 and .3 of RIS]{$\boldsymbol{\Phi}$};
	
	\draw[-to,shorten >=3pt] (BS) to node[above left, pos=.1, yshift=-.7cm] {$\mathbf{V}$} (UE1);
	\draw[-to,shorten >=3pt] (BS) to node[above,pos=.3, ] {$\mathbf{D}$} (UE2);
	
	\draw[-to] (BS) to node[above left, pos=.6] {$\mathbf{H}$} (RIS);
	
	\draw[-to,shorten >=3pt] (RIS) to node[left=-1mm, below=0mm, pos=.3] {$\mathbf{G}$} (UE1);
	\draw[-to,shorten >=3pt] (RIS) to node[below=-3mm, left=2mm] {} (UE2);
\end{tikzpicture}

%% file: problem.tex
\section{Problem Formulation}
\label{sec:problem}

With the derived channel model from \autoref{sec:channel},
we can formulate the communication system optimization problems with \gls{sdma},
which uses the same resource block to serve multiple users.
Therefore, it realizes a higher spectrum efficiency
compared to earlier \gls{ma} techniques,
such as \gls{tdma}, \gls{fdma} and \gls{cdma}~\cite{mao2022rate}.
Moreover,
\gls{sdma} gains popularity due to the high spatial resolution of modern massive multi-antenna systems~\cite{joham2005linear}.
Therefore,
we choose \gls{sdma} a high-performance and future-oriented \gls{ma} technique.
In \gls{sdma},
the interference from other users is considered as noise and
the objective is to maximize the \gls{sinr}.
we define $\mathbf{L}=\mathbf{C}\mathbf{V}$
with $\mathbf{L} \in \mathbb{C}^{U\times U}$.
The \gls{sinr} of user~$u$ is computed as
$|l_{uu}|^2/(\sum_{v\neq u}|l_{uv}|^2+\sigma^2)$,
where $l_{uv}$ is the element in row~$u$ and column~$v$ of $\mathbf{L}$
and $\sigma^2$ is the noise power.
Following the canonical problem formulation of \gls{sdma}~\cite{shi2011iteratively},
we aim to maximize the \gls{wsr} of all users.
The problem is formulated as
\begin{subequations}
\begin{alignat}{2}%
    \max_{\mathbf{V}, \boldsymbol{\Phi}}\quad & %
  f=\sum_{u=1}^U w_u\log_2\left(1+\frac{|l_{uu}|^2}{\sum_{v\neq u}|l_{uv}|^2+\sigma^2}\right) \label{eq:sdma_objective}
\\
    \text{s.t.}\quad & \trace{\left(\mathbf{V}\mathbf{V}^H\right)} \leq E_{Tr}, \label{eq:sdma_transmit_power}\\
    & |\phi_{nn}|=1 \text{ for } n=1, \dots, N, \label{eq:sdma_passive_ris}\\
    & |\phi_{nn'}|=0 \text{ for } n, n' = 1, \dots, N \text{ and }  n \neq n', \label{eq:sdma_offdiagonal}
\end{alignat}
\label{eq:problem_sdma}%
\end{subequations}
\noindent where
$w_u$ in \eqref{eq:sdma_objective} is the weight of user~$u$,
\eqref{eq:sdma_transmit_power} states that the total transmit power cannot exceed the maximum transmit power $E_{Tr}$,
$\phi_{nn}$ in~\eqref{eq:sdma_passive_ris} is the diagonal element in row~$n$ and column~$n$ of $\boldsymbol{\Phi}$.
This constraint ensures that the \gls{ris} does not amplify the received signal (i.e., passive \gls{ris}).
Constraint~\eqref{eq:sdma_offdiagonal} enforces that $\boldsymbol{\Phi}$ is diagonal.
Note that $l_{uv}$ depends on both $\mathbf{V}$ and $\boldsymbol{\Phi}$.
Therefore, both $\mathbf{V}$ and $\boldsymbol{\Phi}$ are the optimization variables.

\begin{remark}
\label{re:rank}
The maximal number of served users is the rank of channel $\mathbf{C}$.
If the direct channel $\mathbf{D}$ is weak,
the rank of $\mathbf{C}$ depends mainly on the rank of 
$\mathbf{G}(\mathbf{I} - \boldsymbol{\Phi}\mathbf{S}_{II})^{-1}\boldsymbol{\Phi}\mathbf{H}$.
Since the rank of the product of matrices is smaller than or equal to the lowest rank of the factors,
the rank of $\mathbf{C}$ depends strongly on ranks of $\mathbf{G}$ and $\mathbf{H}$.
If they are rank-deficient,
it is impossible to serve as many users as the \gls{bs} antenna numbers.
\end{remark}

In the following sections,
we propose the unsupervised \gls{ml} approach to solve Problem~\eqref{eq:problem_sdma}.

\begin{remark}
In \autoref{sec:risnet},
two approaches with full \gls{csi} and partial \gls{csi}
are proposed.
They share the same problem formulation~\eqref{eq:problem_sdma}.
The only difference is the available information:
With full \gls{csi} as the algorithm input,
the objective~\eqref{eq:sdma_objective} can be computed in a deterministic way
and the problem is easier to solve.
However, the full \gls{csi} is difficult to obtain.
With partial \gls{csi},
the information is incomplete to calculate the objective,
because the \gls{csi} computation requires the full \gls{csi},
and the problem is more difficult to solve.
However, the partial \gls{csi} is easier to obtain 
given limited resource for channel estimation,
as will be explained in more detail in \autoref{sec:risnet_partial}.
\end{remark}

%% file: risnet.tex
\section{Unsupervised Machine Learning with RISnet}
\label{sec:risnet}

\input{framework}
\input{architecture}
\input{joint}

\input{algorithm}

%% file: framework.tex
\subsection{The Framework of Unsupervised ML for Optimization}
\label{sec:framework}

We first present the framework of unsupervised \gls{ml} for optimization.
Given a problem representation~$\boldsymbol{\Gamma}$ (in our case, \gls{csi} and user weights~$w_u$ in~\eqref{eq:sdma_objective}),
we look for a solution~$\boldsymbol{\Phi}$ (the \gls{ris} phase shifts) that maximizes objective~$f$ in~\eqref{eq:sdma_objective},
which is fully determined by $\boldsymbol{\Gamma}$ and $\boldsymbol{\Phi}$, and it can be written as
$f(\boldsymbol{\Gamma}, \boldsymbol{\Phi})$.
We define an \gls{nn}~$N_\theta$, which is parameterized by~$\theta$ (i.e., $\theta$ contains all trainable weights and biases in~$N_\theta$) and
maps from~$\boldsymbol{\Gamma}$
to~$\boldsymbol{\Phi}$,
i.e., $\boldsymbol{\Phi} = N_\theta(\boldsymbol{\Gamma})$.
We write the objective as ${f(\boldsymbol{\Gamma}, \boldsymbol{\Phi})=f(\boldsymbol{\Gamma}, N_\theta(\boldsymbol{\Gamma}); \theta)}$.
Note that it is emphasized that~$f$ depends on~$\theta$.
We then collect massive data of $\boldsymbol{\Gamma}$ in a training set~$\mathcal{D}$
and formulate the problem as
\begin{equation}
    \max_\theta K=\sum_{\boldsymbol{\Gamma} \in \mathcal{D}} f(\boldsymbol{\Gamma}, N_\theta(\boldsymbol{\Gamma}); \theta).
    \label{eq:ml}
\end{equation}
In this way, $N_\theta$ is optimized for the emsemble of $\boldsymbol{\Gamma} \in \mathcal{D}$
(\emph{training}) using gradient ascent:
\begin{equation}
    \theta \leftarrow \theta + \eta \nabla_\theta K,
    \label{eq:gradient-ascend}
\end{equation}
where $\eta$ is the learning rate.
If $N_\theta$ is well trained,
$\boldsymbol{\Phi}' = N_\theta (\mathbf{\Gamma'})$ is also a good solution for $\boldsymbol{\Gamma}' \notin \mathcal{D}$ (\emph{testing}),
like a human uses experience to solve new problems of the same type\footnote{A complete retraining is only required when the input states are fundamentally changed, e.g., change of deployment environment.}~\cite{yu2022role}.

Although~\eqref{eq:ml} is a general approach,
it would benefit from the problem-specific domain knowledge.
In the following sections,
we first define the features.
Next, we propose the RISnet architecture, and
finally, we present the joint optimization of \gls{bs} precoding and \gls{ris} configuration.

%% file: architecture.tex
\subsection{Channel Estimation and Feature Definition}
\label{sec:feature_definition}

To begin with the \gls{ml} approach,
we first define the features as input of RISnet.
As depicted in \autoref{fig:system_model},
there are three channel matrices~$\mathbf{H}$, $\mathbf{G}$ and $\mathbf{D}$,
among which
$\mathbf{H}$ is assumed to be constant because \gls{bs} and \gls{ris} are stationary and the environment is relatively invariant,
$\mathbf{G}$ and $\mathbf{D}$ depend on the user positions.
Therefore, they need to be estimated and used as input of $N_\theta$. 
To estimate channel matrix~$\mathbf{G}$,
user~$u$ transmits a pilot signal~$\rho_u$,
which is known to the \gls{ris}.
The received pilot signal is $v_{un} = g_{un} \rho + t$,
where
$g_{un}$ is the channel gain between user~$u$ and \gls{ris} element~$n$,
and $t$ is the thermal noise.
The estimated value of $g_{un}$ is therefore $v_{un} / \rho_u$.
Note that we first assume that every \gls{ris} element has the ability to estimate the channels in
\autoref{sec:pi_risnet}.
This assumption is good for the optimization but
requires expensive hardware.
In \autoref{sec:risnet_partial},
we assume that only a few \gls{ris} elements
are equipped with hardware for channel estimation
with the pilot signals from users,
which requires significantly less complicated hardware,
but sets a more difficult challenge for the optimization.
The estimation of the channel matrix~$\mathbf{D}$ is less challenging 
since the \gls{bs} antennas are significantly less than the \gls{ris} elements.
We can use a channel estimation method described in~\cite{guo2020weighted}
to estimate~$\mathbf{D}$.
With the estimated channel matrices,
we would like to define a feature~$\boldsymbol{\gamma}_{un}$ for user~$u$ and \gls{ris} element~$n$.
Since $g_{un}$ in row~$u$ and column~$n$ of $\mathbf{G}$ is the channel gain from \gls{ris} element~$n$ to user~$u$,
we can simply include amplitude and phase of $g_{un}$ in $\boldsymbol{\gamma}_{un}$\footnote{The \gls{nn} does not take complex numbers as input.}.
On the other hand, elements in $\mathbf{D}$ cannot be mapped to \gls{ris} elements because $\mathbf{D}$ is the channel from the \gls{bs} directly to the users.
Therefore, we define $\mathbf{J}=\mathbf{D}\mathbf{H}^+$,
and \eqref{eq:channel_mutual_coupling} becomes
$\mathbf{y} = \left(\mathbf{G} (\mathbf{I}-\boldsymbol{\Phi}\mathbf{S}_{II})^{-1} \boldsymbol{\Phi} + \mathbf{J}\right) \mathbf{H} \mathbf{V} \mathbf{x} + \mathbf{n}$,
i.e.,
signal~$\mathbf{x}$ is precoded with $\mathbf{V}$, transmitted through channel~$\mathbf{H}$ to the \gls{ris},
and through channel~${\mathbf{G}(\mathbf{I}-\boldsymbol{\Phi}\mathbf{S}_{II})^{-1}\boldsymbol{\Phi} + \mathbf{J}}$
to the users. 
Element~$j_{un}$ of $\mathbf{J}$ can be interpreted as the channel gain from \gls{ris} element~$n$ to user~$u$.
The channel feature of user~$u$ and \gls{ris} element~$n$ can then be defined as
$\boldsymbol{\gamma}_{un} = (|g_{un}|, \arg(g_{un}), |j_{un}|, \arg(j_{un}))^T \in \mathbb{R}^{4 \times 1}$.
The complete feature of all \gls{ris} elements and users is the aggregation of $\boldsymbol{\gamma}_{un}$ for all $u$ and $n$
and user weights.
The concrete structure is described in the following \autoref{sec:risnet_architecture}.

\subsection{The RISnet Architecture}
\label{sec:risnet_architecture}

We design a specific \gls{nn} architecture for \gls{ris} configuration
according to our domain knowledge in wireless communication.
Observing~\eqref{eq:channel_mutual_coupling},
we notice that the optimal phase shift of every \gls{ris} element
depends on its own channel gain
and a common goal of the whole \gls{ris},
which should be shared among all \gls{ris} elements
to enable their cooperation.
Correspondingly,
we define a \emph{local feature} of every \gls{ris} element
and a common \emph{global feature} of all \gls{ris} elements.
The RISnet consists of $L$~layers.
In each layer,
we design information processing units to generate local features
and global features for every \gls{ris} element and user,
which are stacked as the input of the next layer,
such that a proper information flow can be created for a sophisticated decision on~$\boldsymbol{\Phi}$.
The idea is comparable to using convolutional layers to capture local pattern in computer vision
and using attention mechanism to model context in nature language processing,
but for a much more specific problem with many \gls{ris} elements serving multiple users together.

By using the same information processing units for all \gls{ris} elements,
the number of trainable parameters is independent of the number of \gls{ris} elements.
With this approach,
RISnet can configure more than 1000~elements with an adequate complexity,
allowing for low complexity in training and
high efficiency in inference (application).

Another important consideration
from domain knowledge in communication
is the \emph{permutation-invariance}.
From~\eqref{eq:sdma_objective},
we notice that a permutation of the users does not have an impact on the objective function in \gls{sdma}.
The optimal decision on $\boldsymbol{\Phi}$ should therefore be independent from the user order.
This property is called permutation-invariance.
We define a permutation matrix~$\mathbf{P} \in \{0, 1\}^{U \times U}$ to describe an arbitrary permutation,
where each row and each column has only one 1.
For example, let us define
\begin{equation*}
\mathbf{P}=
\begin{pmatrix}%
1 & 0 & 0\\
0 & 0 & 1\\
0 & 1 & 0\\
\end{pmatrix}.
\end{equation*}
$\mathbf{P} \mathbf{G}$ permutes the second and third row of $\mathbf{G}$ while the first row remains unchanged.

\begin{definition}
A neural network~$N_\theta$ is permutation-invariant if ${N_\theta(\mathbf{P}\boldsymbol{\Gamma}) = N_\theta(\boldsymbol{\Gamma})}$ for any permutation matrix~$\mathbf{P}$.
\end{definition}

It is desirable that RISnet is permutation-invariant for \eqref{eq:problem_sdma}
to reflect the nature of \gls{sdma}.

A third consideration for RISnet design is that even if RISnet has a good scalability for more than one thousand \gls{ris} elements,
the full \gls{csi} of all \gls{ris} elements is extremely difficult to acquire.
Therefore, it is very beneficial to have an input,
which is easier to acquire than the full \gls{csi}.

\subsubsection{RISnet Architecture with Full CSI}
\label{sec:pi_risnet}

In the following,
we present the RISnet architecture with the above-described high scalability \gls{wrt} \gls{ris} elements,
permutation-invariance
and low requirement for \gls{csi} input.
The RISnet has multiple layers.
Both input and output of a layer are three-dimensional tensors,
where
the first dimension is the feature,
the second dimension is the \gls{ris} element
and the third dimension is the user.
In the first layer,
the vector~$\mathbf{f}_{un,1}=\boldsymbol{\Gamma}[\cdot, n, u]$ is the feature of \gls{ris} element~$n$
and user~$u$,
defined as the concatenation of user weight~$w_u$ and channel feature~$\boldsymbol{\gamma}_{un}$ (defined in \autoref{sec:feature_definition}).
The input and output feature format is 
shown in \autoref{fig:info_processing}.

As described at the beginning of this section,
the decision on the optimal phase shift of every \gls{ris} element depends on
both the local feature of the current \gls{ris} element
and the global feature of the whole \gls{ris}.
Therefore,
for each \gls{ris} element and user,
we define 4 classes of information processing units:
\begin{itemize}
    \item \texttt{c}urrent user and \texttt{c}urrent \gls{ris} element (\texttt{cc}),
    \item \texttt{c}urrent user and \texttt{a}ll \gls{ris} elements (\texttt{ca}),
    \item \texttt{o}ther users and \texttt{c}urrent \gls{ris} element (\texttt{oc}),
    \item \texttt{o}ther users and \texttt{a}ll \gls{ris} elements (\texttt{oa}).
\end{itemize}
Denote the input feature of user~$u$ and \gls{ris} element~$n$ in layer~$i$ as $\mathbf{f}_{un,i}$,
the output feature of user~$u$ and \gls{ris} element~$n$ in layer~$i$ is calculated as
\begin{multline}
\mathbf{f}_{un, i + 1} =\\
\begin{pmatrix}
     \text{ReLU}(\mathbf{W}^{\texttt{cc}}_{i} \mathbf{f}_{un, i} + \mathbf{b}_i^{\texttt{cc}}) \\
     \left(\sum_{n'}\text{ReLU}(\mathbf{W}^{\texttt{ca}}_{i} \mathbf{f}_{un', i} + \mathbf{b}_i^{\texttt{ca}})\right) \big/ N\\
     \left(\sum_{u'\neq u}\text{ReLU}(\mathbf{W}^{\texttt{oc}}_{i} \mathbf{f}_{u'n, i} + \mathbf{b}_i^{\texttt{oc}})\right) \big/ (U-1)\\
     \left(\sum_{u'\neq u}\sum_{n'}\text{ReLU}(\mathbf{W}^{\texttt{oa}}_{i} \mathbf{f}_{u'n', i} + \mathbf{b}_i^{\texttt{oa}})\right) \big/ (N(U-1)) %
\end{pmatrix}
\label{eq:layer_processing}
\end{multline}
for $i < L$,
where 
${\mathbf{W}^{\texttt{cc}}_i \in \mathbb{R}^{Q_i \times P_i}}$ is a matrix with trainable weights of class~\texttt{cc} in layer~$i$ with the input feature dimension~$P_i$ in layer~$i$ (i.e., ${\mathbf{f}_{un, i} \in \mathbb{R}^{P_i \times 1}}$) and output feature dimension~$Q_i$ in layer~$i$ of class~\texttt{cc},
${\mathbf{b}^{\texttt{cc}}_i \in \mathbb{R}^{Q_i \times 1}}$ is trainable bias of class~\texttt{cc} in layer~$i$.
Similar definitions and same dimensions apply to classes~\texttt{ca}, \texttt{oc} and \texttt{oa}.

For class \texttt{cc} in layer~$i$
(the first line of \eqref{eq:layer_processing}),
the output feature of user~$u$ and \gls{ris} element~$n$ is computed by 
applying a conventional fully connected layer 
(a linear transform with weights $\mathbf{W}^{\texttt{cc}}_i$ and bias $\mathbf{b}^{\texttt{cc}}_i$ and the ReLU activation)
to input $\mathbf{f}_{un,i}$.
The difference to the conventional fully connected \gls{nn}
is that the information processing is applied to
feature of every \gls{ris} element and user individually,
instead of the whole input of all \gls{ris} elements and users.

For class \texttt{ca} in layer~$i$
(the second line of \eqref{eq:layer_processing}),
we first apply the conventional linear transform
and ReLU activation to $\mathbf{f}_{un',i}$ like class \texttt{cc},
where $n=1, \dots, N$,
then compute the mean value of all \gls{ris} elements.
Therefore, the output feature of class \texttt{ca} for user~$u$ and all \gls{ris} elements is the same.

For classes \texttt{oc} and \texttt{oa}
(the third and fourth lines of \eqref{eq:layer_processing}),
the output features are averaged over all elements and/or other users,
similar to class \texttt{ca}.

We can infer from the above description that $\mathbf{f}_{un, i + 1} \in \mathbb{R}^{4Q_i}$ for all $u$ and $n$
because the output feature comprises of four classes.
Therefore $P_{i + 1} = 4 Q_i$.
The whole output feature
$\mathbf{F}_{i + 1} \in \mathbb{R}^{4Q_i \times U \times N}$ is a three dimensional tensor,
where elements with index $u$ and $n$ in second and third dimensions are $\mathbf{f}_{un, i + 1}$.
Observe~\eqref{eq:layer_processing},
we note that the same information processing units are applied to all users and \gls{ris} elements.
Therefore,
the number of trainable parameters is independent from the number of \gls{ris} elements,
which enables a high scalability to configure more than 1000 \gls{ris} elements.
For the final layer,
we use one information processing unit.
Features of different users are summed up to be the phase shifts because all the users share the same phase shift.
Element~$\phi_{nn}$ in row~$n$ and column~$n$ of $\boldsymbol{\Phi}$ is defined as
$\phi_{nn}=\e^{\jj\varphi_n}$,
where $\varphi_n$ is the $n$-th output of RISnet.
Since $|\e^{\jj\varphi}|=1$ for any $\varphi$,
constraint~\eqref{eq:sdma_passive_ris} is satisfied.
Since all off-diagonal elements of $\boldsymbol{\Phi}$ are initialized as 0
and stay constant during the optimization,
constraint~\eqref{eq:sdma_offdiagonal} is satisfied.
The information processing of a layer is illustrated in \autoref{fig:info_processing}.

Another important merit of RISnet for \gls{sdma} is the permutation-invariance.
\begin{figure}%
    \centering
    \subfigure[First layer]{\resizebox{.9\linewidth}{!}{\input{figs/layer1.tex}}}
    \subfigure[Intermediate layers]{\resizebox{.9\linewidth}{!}{\input{figs/layer2}}}
    \subfigure[Final layer]{\resizebox{.9\linewidth}{!}{\input{figs/layer3}}}
    \caption{Information processing (info. proc.) in RISnet.
    The symmetric information processing along the dimension of users makes RISnet invariant to user permutation.}
    \label{fig:info_processing}
\end{figure}
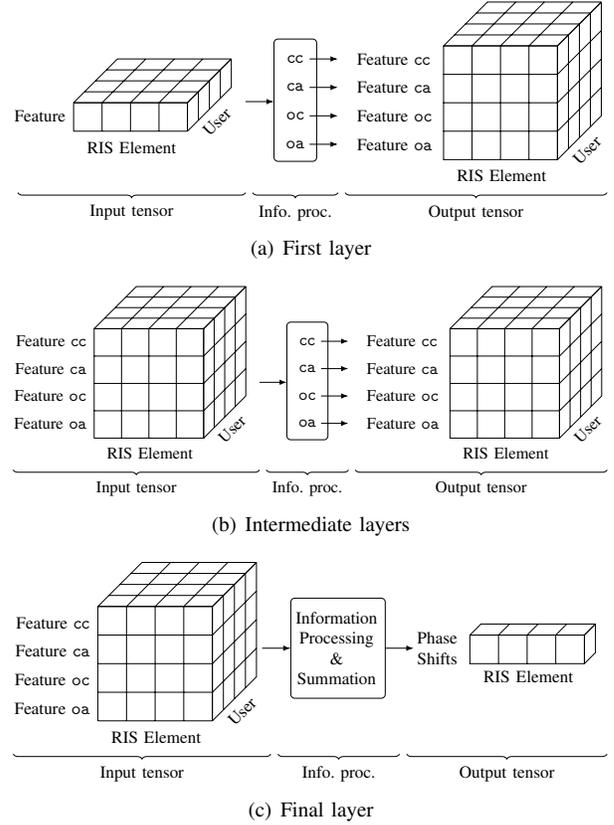

\begin{theorem}
The RISnet is permutation-invariant.
\label{theorem:permutation_invariance}
\end{theorem}

\begin{proof}
The permutation-invariance can be intuitively achieved in the following way:
Feature \texttt{cc} and Feature \texttt{ca} (current user and current/all \gls{ris} element(s))
depend only on the input feature of current user
and are independent from
user order.
Feature \texttt{oc} and Feature \texttt{oa} (other users and current/all \gls{ris} element(s))
depend only on the input feature of other users
and are also independent from user order.
Since the summation of features over all users
in the last layer of RISnet is permutation-invariant,
the RISnet is permutation-invariant.
A rigid proof is available in \autoref{appendix}.
\end{proof}

\begin{remark}
RISnet can be considered as a highly specialized \gls{gnn},
which performs inference on a graph with vertices and edges between vertices.
In our case,
a vertex is a combination of \gls{ris} element and user.
The class~\texttt{cc} is the local feature of the vertex.
Three sets of edges represent classes~\texttt{ca}, \texttt{oc} and~\texttt{oa},
where edges connect all \gls{ris} elements of the same user,
users of the same \gls{ris} element,
and all \gls{ris} elements and users,
respectively.
The messages passed from neighboring vertices
in classes~\texttt{ca}, \texttt{oc} and~\texttt{oa}
are described by lines two, three and four of~\eqref{eq:layer_processing},
respectively.
The aggregation of messages from neighboring vertices is the averaging operations in~\eqref{eq:layer_processing}.
The local feature and the global features are combined in the feature dimension as the
input for the next layer.
\end{remark}

\subsubsection{RISnet Architecture with Partial CSI}
\label{sec:risnet_partial}

Although the above-presented architecture has a high scalability due to the reuse of information processing units for all \gls{ris} elements and users,
the channel estimation is still difficult.
In particular,
the estimation of $\mathbf{G}$ is especially challenging due to
the large number of \gls{ris} elements
(unlike $\mathbf{D}$)
and high variance since it depends on the user location
(unlike $\mathbf{H}$).
If every \gls{ris} element has the ability to estimate \gls{csi}
from the pilot signal,
the hardware would be very complex.
However,
using the domain knowledge in channel modeling,
we notice that
if the propagation paths in the channel are mostly \gls{los} or specular,
i.e., the channel is \enquote{sparse},
then the \gls{csi} of a few \gls{ris} elements contains sufficient information about the user location,
which can be used to infer the full \gls{csi} of all \gls{ris} elements.
This fact suggests that we can use \emph{partial \gls{csi}} of $\mathbf{G}$ as input to RISnet.
The partial \gls{csi} is defined as channel gains between \gls{ris} and user
of a few selected \gls{ris} elements equipped with hardware for channel estimation,
i.e., a few selected columns of $\mathbf{G}$.
The partial \gls{csi} can be estimated with received pilot signals from users.
Since only a few \gls{ris} elements have hardware
for channel estimation,
the whole \gls{ris} has a low hardware complexity.
Instead of performing an explicit full \gls{csi} prediction like an image super-resolution,
we perform an implicit full \gls{csi} prediction,
i.e., an end-to-end learning from partial \gls{csi} to a complete \gls{ris} configuration.
In this section,
we propose an \gls{ris} architecture where only a few \gls{ris} elements are equipped with hardware to estimate the
\gls{csi} with pilot signals from the users,
as described in \autoref{sec:feature_definition}.
Compared to an \gls{ris} which can estimate full \gls{csi},
the proposed hardware architecture
has a significantly lower complexity due to the small number of \gls{ris} elements with channel estimation ability.
This hardware structure is similar to hybrid \gls{ris}~\cite{ju2024beamforming},
in which a few active elements with RF chains can amplify signals,
but is simpler for implementation
because the elements only estimate the channel rather than amplifying the signal,
which justifies the feasibility of the proposed hardware architecture.

\Gls{ris} elements with the hardware of channel estimation from received pilot signals of users
are defined as \emph{anchor elements}.
They are uniformly placed on the \gls{ris},
because we would like to make the partial \gls{csi} as representative to the full \gls{csi} as possible.
The geometric consideration will become clear in the following description.

The input channel features in this section consist of features of all users and anchor elements.
The RISnet expands from the anchor elements to all \gls{ris} elements.
A layer in RISnet that expands from anchor elements is called an \emph{expansion layer}.
The basic idea of the expansion layer is to apply the same information processing unit to an adjacent \gls{ris} element with the same relative position to the anchor element,
as shown in \autoref{fig:expansion_filter},
where element~5 is an anchor element.
The expansion layer computes features of all 9 elements with only feature of element~5
($\mathbf{f}_5$).
Therefore, each class, i.e., \texttt{cc}, \texttt{ca}, \texttt{oc}, \texttt{oa} for \gls{sdma} and \texttt{c}, \texttt{a} for \gls{noma}, has 9~information processing units,
which outputs features of the same \gls{ris} element and the adjacent 8~elements.

Concretely,
the output of \gls{ris} element~$n$ using information processing unit~$j$ is computed as
\begin{multline}
\mathbf{f}_{u\nu(n, j), i + 1} = \\
\begin{pmatrix}
     \text{ReLU}(\mathbf{W}^{\texttt{cc}}_{i,j} \mathbf{f}_{un, i} + \mathbf{b}_{i,j}^{\texttt{cc}}) \\
     \left(\sum_{n'}\text{ReLU}(\mathbf{W}^{\texttt{co}}_{i,j} \mathbf{f}_{un', i} + \mathbf{b}_{i,j}^{\texttt{co}})\right) \big/ N\\
     \left(\sum_{u'\neq u}\text{ReLU}(\mathbf{W}^{\texttt{oc}}_{i,j} \mathbf{f}_{u'n, i} + \mathbf{b}_{i,j}^{\texttt{oc}})\right) \big/ (U-1)\\
     \left(\sum_{u'\neq u}\sum_{n'}\text{ReLU}(\mathbf{W}^{\texttt{oa}}_{i,j} \mathbf{f}_{u'n', i} + \mathbf{b}_{i,j}^{\texttt{oa}})\right) \big/ (N(U-1)) %
\end{pmatrix},
\label{eq:expansion_layer_processing}
\end{multline}
where $\nu(n, j)$ is the \gls{ris} element index when applying information processing unit~$j$ for input of \gls{ris} element~$n$.
According to \autoref{fig:expansion_filter} and assuming that the \gls{ris} element index begins with~1 at the upper left corner,
increases first along the row and then changes to the next row 
(i.e., the index in row~$w$ and column~$h$ is ${h + (w - 1)\cdot H}$, with $H$ being the number of columns of the \gls{ris} array),
we have
\begin{equation}
    \nu(n, j)= 
\begin{cases}
n - H - 2 + j & j = 1, 2, 3,\\
n - 5 + j & j = 4, 5, 6,\\
n + H - 8 + j & j = 7, 8, 9.\\
\end{cases}
\label{eq:nu}
\end{equation}
In~\eqref{eq:expansion_layer_processing},
we use the four information processing units defined in~\eqref{eq:layer_processing}
to process feature of \gls{ris} element~$n$ and user~$u$.
Unlike in~\eqref{eq:layer_processing},
where the output is feature of \gls{ris} element~$n$ and user~$u$,
the output in~\eqref{eq:expansion_layer_processing}
is feature of \gls{ris} element~$\nu(n, j)$ and user~$u$.
For example,
for $j=2$,
the output feature is for the element above element~$n$
(see~\eqref{eq:nu}).
In this way,
we generate features of 9~\gls{ris} elements out of the feature of 1~element.

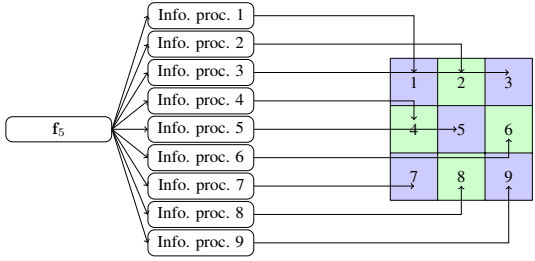
\begin{figure}%
    \centering
    \resizebox{.8\linewidth}{!}{\input{figs/partial_risnet}}
    \caption{Application of 9 information processing units to expand from one anchor RIS element to 9 RIS elements,
    where $\mathbf{f}_5$ is the channel feature of \gls{ris} element 5.
    Information processing unit~5 is comparable to the information processing units of RISnet with full \gls{csi}.
    Indices of user and layer are omitted for simplicity since the expansion is for \gls{ris} elements.}
    \label{fig:expansion_filter}
\end{figure}

By defining two such expansion layers,
the numbers of anchor elements are increased by a factor of 9 in both row and column.
If we have 16 anchor elements (4$\times$4) with \gls{csi},
we can generate phase shifts of 1296 (36$\times$36) \gls{ris} elements. %
The process of expanding anchor elements is illustrated in \autoref{fig:expansion},
where the blue \gls{ris} elements in \autoref{fig:expansion-partial1} can estimate the channel from the pilot signals from the users.
Such \gls{ris} elements are only 1/81 of all \gls{ris} elements.
The whole RISnet architecture is shown in \autoref{fig:risnet_structure}.

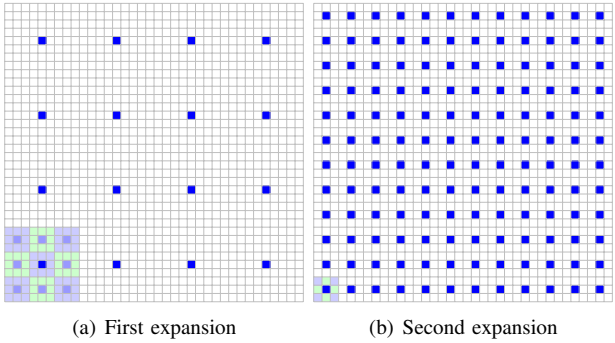
\begin{figure}
    \centering
    \subfigure[First expansion\label{fig:expansion-partial1}]{\resizebox{.45\linewidth}{!}{\input{figs/partial1}}}
    \subfigure[Second expansion\label{fig:expansion-partial2}]{\resizebox{.45\linewidth}{!}{\input{figs/partial2}}}
    \caption{Expansion of considered RIS elements. Blue: anchor RIS elements.
    Lower left corner: example of the expansion to extend the anchor RIS elements from the blue element to the adjacent elements (light blue elements in Subfigure~(a) and all elements in Subfigure~(b)).}
    \label{fig:expansion}
\end{figure}

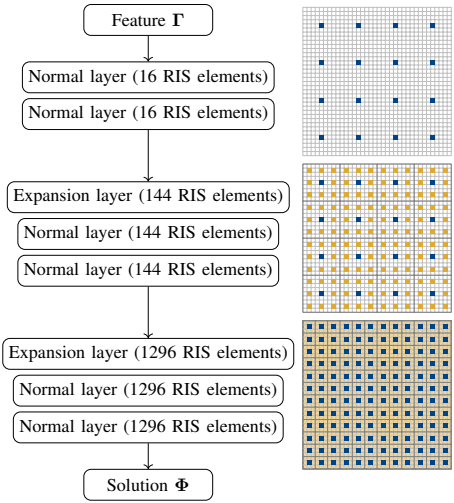
\begin{figure}%
    \centering
    \resizebox{!}{.27\textheight}{\input{figs/layer_structure-alternative.tex}}
    \caption{The RISnet architecture with partial \gls{csi},
    where the information processing of normal layers is given by \eqref{eq:layer_processing} 
    and the information processing of expansion layers is given by \eqref{eq:expansion_layer_processing}. 
    Note that this process is only possible with uniformly placed anchor elements.}
    \label{fig:risnet_structure}
\end{figure}

\begin{remark}
The \gls{ris} elements with the channel estimation capability are fixed once the \gls{ris} hardware is designed.
Since the hardware design does not change constantly,
we assume known and fixed \gls{ris} elements with \gls{csi}.
If the \gls{ris} layout is modified,
the \gls{nn} architecture must be modified accordingly.
For example, if we want that one element is expanded to the adjacent 36~elements instead of 9,
we should define 36~information processing units (see \autoref{fig:expansion}).
This is an example of problem-specific \gls{ml}
that the \gls{nn} architecture is determined according to the hardware structure.
\end{remark}

During the model training,
the full \gls{csi} is still required to compute the objectives~\eqref{eq:sdma_objective}\footnote{The full \gls{csi} can be obtained, e.g., by off-line channel measurement before the operation.}.
However, phase shifts of all \gls{ris} elements~$\boldsymbol{\Phi}$ are computed with the partial \gls{csi} after the training in the application,
which implies that only the partial \gls{csi}
on the anchor elements is required in application.

\subsubsection{Efficient Parallel Implementation with Tensor Operations}

Although~\eqref{eq:layer_processing} and~\eqref{eq:expansion_layer_processing}
are the most intuitive way to understand the information flow in RISnet,
the computation is done per \gls{ris} element (and user),
which can only be implemented in a loop and has a low computation efficiency.
In order to utilize the parallel computing in a GPU,
it would be desirable to implement the information processing as tensor operations instead of computation per \gls{ris} element (and user).
Let $\mathbf{F}^{ca}_{i+1}$ be the output of information processing unit~\texttt{ca} in layer~$i$,
the second row of the right hand side of \eqref{eq:layer_processing} can be rewritten as
\begin{equation}
    \mathbf{F}^{\texttt{ca}}_{i+1}[\cdot, n, u]=\left(\sum_{n'}\text{ReLU}(\mathbf{W}^{\texttt{ca}}_{i} \mathbf{f}_{un', i} + \mathbf{b}_i^{\texttt{ca}})\right) / N.
    \label{eq:tensor_operation1}
\end{equation}
It is easy to prove that \eqref{eq:tensor_operation1}~is equivalent to
\begin{equation}
    \mathbf{F}^{\texttt{ca}}_{i+1}=\text{ReLU}(\mathbf{W}^{\texttt{ca}}_{i} \mathbf{F}_{i} + \mathbf{b}_i^{\texttt{ca}}) \cdot \mathbf{1}^{N\times N} / N,
    \label{eq:tensor_operation2}
\end{equation}
where the multiplication between $\mathbf{W}^{\texttt{ca}}_{i}$ and $\mathbf{F}_{i}$ is done in the first dimension (the feature dimension) of $\mathbf{F}_i$
and
the multiplication between $\text{ReLU}(\mathbf{W}^{\texttt{ca}}_{i} \mathbf{F}_{i} + \mathbf{b}_i^{\texttt{ca}})$ and $\mathbf{1}^{N\times N}$ takes place in the last two dimensions (users and \gls{ris} elements) of $\text{ReLU}(\mathbf{W}^{\texttt{ca}}_{i} \mathbf{F}_{i} + \mathbf{b}_i^{\texttt{ca}})$.
In this way,
\eqref{eq:tensor_operation1} is computed for all \gls{ris} elements and users without a loop using tensor operations.
Similarly,
the third row of the right hand side of \eqref{eq:layer_processing} can be rewritten as
\begin{equation}
    \mathbf{F}^{ \texttt{oc}}_{i+1}[\cdot, n, u] = \!\!
    \left(\sum_{u'\neq u}\text{ReLU}(\mathbf{W}^{\texttt{oc}}_{i} \mathbf{f}_{u'n, i} + \mathbf{b}_i^{\texttt{oc}})\right) \big/ (U-1).
    \label{eq:tensor_operation3}
\end{equation}
It is straightforward to prove that \eqref{eq:tensor_operation3} is equivalent to
\begin{equation}
    \mathbf{F}^{ \texttt{oc}}_{i+1} = \mathbf{E}^{U \times U} \cdot \text{ReLU}(\mathbf{W}^{\texttt{oc}}_{i} \mathbf{F}_{i} + \mathbf{b}_i^{\texttt{oc}}) \big/ (U-1).
    \label{eq:tensor_operation4}
\end{equation}
Other operations in~\eqref{eq:layer_processing}--\eqref{eq:expansion_layer_processing} can also be parallelized in similar ways.
Therefore, the RISnet can be implemented in a computation-efficient way.

%% file: figs/layer1.tex
\begin{tikzpicture}
\makeatletter %
\tikzoption{canvas is xy plane at z}[]{%
	\def\tikz@plane@origin{\pgfpointxyz{0}{0}{#1}}%
	\def\tikz@plane@x{\pgfpointxyz{1}{0}{#1}}%
	\def\tikz@plane@y{\pgfpointxyz{0}{1}{#1}}%
	\tikz@canvas@is@plane
}
\makeatother

\NewDocumentCommand{\DrawCubes}{O {} m m m m m m}{%
	\def\XGridMin{#2}
	\def\XGridMax{#3}
	\def\YGridMin{#4}
	\def\YGridMax{#5}
	\def\ZGridMin{#6}
	\def\ZGridMax{#7}
	\begin{scope}[canvas is xy plane at z=\ZGridMax]
		\draw [#1] (\XGridMin,\YGridMin) grid (\XGridMax,\YGridMax);
	\end{scope}
	\begin{scope}[canvas is yz plane at x=\XGridMax]
		\draw [#1] (\YGridMin,\ZGridMin) grid (\YGridMax,\ZGridMax);
	\end{scope}
	\begin{scope}[canvas is xz plane at y=\YGridMax]
		\draw [#1] (\XGridMin,\ZGridMin) grid (\XGridMax,\ZGridMax);
	\end{scope}
}%

\DrawCubes [step=10mm,thick]{0}{4}{1}{2}{0}{4}
\node (width) [rectangle, yshift=-1.1cm, xshift=0.5cm,font=\LARGE] {RIS Element};
\node (height) [rectangle, yshift=0cm, xshift=-1.7cm, anchor=east,font=\LARGE] {Feature};
\node (channels) [rectangle, yshift=-.2cm, xshift=3.5cm,rotate=45,font=\LARGE] {User};
\draw [decorate,decoration={brace,amplitude=5pt,mirror}]
(-3.5,-2.7) -- (4.5,-2.7) node[midway, yshift=-.7cm, font=\LARGE]{Input tensor};

\node[draw,thick,rounded corners,rotate=90,inner sep=0em,minimum width=4.3cm,minimum height=1.5cm] (filters) at (6.25, 0.55) {};
\node[font=\LARGE] at (6.25, 2) {\texttt{cc}};
\node[font=\LARGE] at (6.25, 1) {\texttt{ca}};
\node[font=\LARGE] at (6.25, 0) {\texttt{oc}};
\node[font=\LARGE] at (6.25, -1) {\texttt{oa}};

\draw[-Latex,thick, line width=1pt,] (6.75, 2) -- ++(1, 0);
\draw[-Latex,thick, line width=1pt,] (6.75, 1) -- ++(1, 0);
\draw[-Latex,thick, line width=1pt,] (6.75, 0) -- ++(1, 0);
\draw[-Latex,thick, line width=1pt,] (6.75, -1) -- ++(1, 0);
\draw[-Latex,thick, line width=1pt,] (4.5, 0.5) -- ++(1, 0);
\draw [decorate,decoration={brace,amplitude=5pt,mirror}]
(4.75,-2.7) -- (7.75,-2.7) node[midway, yshift=-.7cm, font=\LARGE]{Info. proc.};

\DrawCubes [step=10mm,thick]{12.99}{17}{0}{4}{0}{4}
\node (width) [yshift=-2.1cm, xshift=13.5cm,font=\LARGE] {RIS Element};
\node (height) [rectangle, yshift=2cm, xshift=8.3cm, anchor=west,font=\LARGE] {Feature \texttt{cc}};
\node (height) [rectangle, yshift=1cm, xshift=8.3cm, anchor=west,font=\LARGE] {Feature \texttt{ca}};
\node (height) [rectangle, yshift=0cm, xshift=8.3cm, anchor=west,font=\LARGE] {Feature \texttt{oc}};
\node (height) [rectangle, yshift=-1cm, xshift=8.3cm, anchor=west,font=\LARGE] {Feature \texttt{oa}};
\node (channels) [rectangle, xshift=4cm, yshift=-1.2cm, xshift=12.5cm,rotate=45,font=\LARGE] {User};
\draw [decorate,decoration={brace,amplitude=5pt,mirror}]
(8,-2.7) -- (17.2,-2.7) node[midway, yshift=-.7cm, font=\LARGE]{Output tensor};
\end{tikzpicture}

%% file: figs/layer2.tex
\begin{tikzpicture}
\makeatletter %
\tikzoption{canvas is xy plane at z}[]{%
	\def\tikz@plane@origin{\pgfpointxyz{0}{0}{#1}}%
	\def\tikz@plane@x{\pgfpointxyz{1}{0}{#1}}%
	\def\tikz@plane@y{\pgfpointxyz{0}{1}{#1}}%
	\tikz@canvas@is@plane
}
\makeatother

\NewDocumentCommand{\DrawCubes}{O {} m m m m m m}{%
	\def\XGridMin{#2}
	\def\XGridMax{#3}
	\def\YGridMin{#4}
	\def\YGridMax{#5}
	\def\ZGridMin{#6}
	\def\ZGridMax{#7}
	\begin{scope}[canvas is xy plane at z=\ZGridMax]
		\draw [#1] (\XGridMin,\YGridMin) grid (\XGridMax,\YGridMax);
	\end{scope}
	\begin{scope}[canvas is yz plane at x=\XGridMax]
		\draw [#1] (\YGridMin,\ZGridMin) grid (\YGridMax,\ZGridMax);
	\end{scope}
	\begin{scope}[canvas is xz plane at y=\YGridMax]
		\draw [#1] (\XGridMin,\ZGridMin) grid (\XGridMax,\ZGridMax);
	\end{scope}
}%

    \DrawCubes [step=10mm,thin]{0}{4}{0}{4}{0}{4}
    \node (width) [rectangle, yshift=-2.1cm, xshift=0.5cm,font=\LARGE] {RIS Element};
    \node (height) [rectangle, yshift=2cm, xshift=-1.7cm, anchor=east,font=\LARGE] {Feature \texttt{cc}};
    \node (height) [rectangle, yshift=1cm, xshift=-1.7cm, anchor=east,font=\LARGE] {Feature \texttt{ca}};
    \node (height) [rectangle, yshift=0cm, xshift=-1.7cm, anchor=east,font=\LARGE] {Feature \texttt{oc}};
    \node (height) [rectangle, yshift=-1cm, xshift=-1.7cm, anchor=east,font=\LARGE] {Feature \texttt{oa}};
    \node (channels) [rectangle, yshift=-1.2cm, xshift=3.5cm,rotate=45,font=\LARGE] {User};

\draw [decorate,decoration={brace,amplitude=5pt,mirror}]
(-4.5,-2.7) -- (4.5,-2.7) node[midway, yshift=-.7cm, font=\LARGE]{Input tensor};

\node[draw,thick,rounded corners,rotate=90,inner sep=0em,minimum width=4.3cm,minimum height=1.5cm] (filters) at (6.25, 0.55) {};
\node[font=\LARGE] at (6.25, 2) {\texttt{cc}};
\node[font=\LARGE] at (6.25, 1) {\texttt{ca}};
\node[font=\LARGE] at (6.25, 0) {\texttt{oc}};
\node[font=\LARGE] at (6.25, -1) {\texttt{oa}};

\draw[-Latex,thick, line width=1pt,] (6.75, 2) -- ++(1, 0);
\draw[-Latex,thick, line width=1pt,] (6.75, 1) -- ++(1, 0);
\draw[-Latex,thick, line width=1pt,] (6.75, 0) -- ++(1, 0);
\draw[-Latex,thick, line width=1pt,] (6.75, -1) -- ++(1, 0);
\draw[-Latex,thick, line width=1pt,] (4.5, 0.5) -- ++(1, 0);
\draw [decorate,decoration={brace,amplitude=5pt,mirror}]
(4.75,-2.7) -- (7.75,-2.7) node[midway, yshift=-.7cm, font=\LARGE]{Info. proc.};

\DrawCubes [step=10mm,thick]{12.99}{17}{0}{4}{0}{4}
\node (width) [yshift=-2.1cm, xshift=13.5cm,font=\LARGE] {RIS Element};
\node (height) [rectangle, yshift=2cm, xshift=8.3cm, anchor=west,font=\LARGE] {Feature \texttt{cc}};
\node (height) [rectangle, yshift=1cm, xshift=8.3cm, anchor=west,font=\LARGE] {Feature \texttt{ca}};
\node (height) [rectangle, yshift=0cm, xshift=8.3cm, anchor=west,font=\LARGE] {Feature \texttt{oc}};
\node (height) [rectangle, yshift=-1cm, xshift=8.3cm, anchor=west,font=\LARGE] {Feature \texttt{oa}};
\node (channels) [rectangle, xshift=4cm, yshift=-1.2cm, xshift=12.5cm,rotate=45,font=\LARGE] {User};
\draw [decorate,decoration={brace,amplitude=5pt,mirror}]
(8,-2.7) -- (17.2,-2.7) node[midway, yshift=-.7cm, font=\LARGE]{Output tensor};
\end{tikzpicture}

%% file: figs/layer3.tex
\begin{tikzpicture}
\makeatletter %
\tikzoption{canvas is xy plane at z}[]{%
	\def\tikz@plane@origin{\pgfpointxyz{0}{0}{#1}}%
	\def\tikz@plane@x{\pgfpointxyz{1}{0}{#1}}%
	\def\tikz@plane@y{\pgfpointxyz{0}{1}{#1}}%
	\tikz@canvas@is@plane
}
\makeatother

\NewDocumentCommand{\DrawCubes}{O {} m m m m m m}{%
	\def\XGridMin{#2}
	\def\XGridMax{#3}
	\def\YGridMin{#4}
	\def\YGridMax{#5}
	\def\ZGridMin{#6}
	\def\ZGridMax{#7}
	\begin{scope}[canvas is xy plane at z=\ZGridMax]
		\draw [#1] (\XGridMin,\YGridMin) grid (\XGridMax,\YGridMax);
	\end{scope}
	\begin{scope}[canvas is yz plane at x=\XGridMax]
		\draw [#1] (\YGridMin,\ZGridMin) grid (\YGridMax,\ZGridMax);
	\end{scope}
	\begin{scope}[canvas is xz plane at y=\YGridMax]
		\draw [#1] (\XGridMin,\ZGridMin) grid (\XGridMax,\ZGridMax);
	\end{scope}
}%

    \DrawCubes [step=10mm,thin]{0}{4}{0}{4}{0}{4}
    \node (width) [rectangle, yshift=-2.1cm, xshift=0.5cm,font=\LARGE] {RIS Element};
    \node (height) [rectangle, yshift=1.9cm, xshift=-1.7cm, anchor=east,font=\LARGE] {Feature \texttt{cc}};
    \node (height) [rectangle, yshift=0.9cm, xshift=-1.7cm, anchor=east,font=\LARGE] {Feature \texttt{ca}};
    \node (height) [rectangle, yshift=-0.1cm, xshift=-1.7cm, anchor=east,font=\LARGE] {Feature \texttt{oc}};
    \node (height) [rectangle, yshift=-1.1cm, xshift=-1.7cm, anchor=east,font=\LARGE] {Feature \texttt{oa}};
    \node (channels) [rectangle, yshift=-1.2cm, xshift=3.5cm,rotate=45,font=\LARGE] {User};

\draw [decorate,decoration={brace,amplitude=5pt,mirror}]
(-4.5,-2.7) -- (4.5,-2.7) node[midway, yshift=-.7cm, font=\LARGE]{Input tensor};

\node[font=\LARGE, align=center,anchor=west] at (5.25, 1) {Information\\Processing\\\&\\Summation};
\draw[black,rounded corners,thick] (5.2, -0.75) rectangle (8.5, 2.75);

\draw[-Latex,thick, line width=1pt,] (8.5, 1) -- ++(1, 0);
\draw[-Latex,thick, line width=1pt,] (4.2, 1) -- ++(1, 0);
\draw [decorate,decoration={brace,amplitude=5pt,mirror}]
(4.75,-2.7) -- (9,-2.7) node[midway, yshift=-.7cm, font=\LARGE]{Info. proc.};

\DrawCubes [step=10mm,thick]{12.99}{17}{2}{3}{3}{4}
\node (width) [yshift=-0cm, xshift=13.5cm,font=\LARGE] {RIS Element};
\node (height) [rectangle, yshift=1cm, xshift=9.5cm, anchor=west,font=\LARGE, align=left] {Phase\\Shifts};
\draw [decorate,decoration={brace,amplitude=5pt,mirror}]
(9.25,-2.7) -- (16.25,-2.7) node[midway, yshift=-.7cm, font=\LARGE]{Output tensor};
\end{tikzpicture}

%% file: figs/partial_risnet.tex
    \begin{tikzpicture}
\tikzstyle{block} = [rectangle, rounded corners, text width=2cm, text centered, draw=black]
    
\node (feature) [block] {$\mathbf{f}_5$};
\node (filter1) [block, right of=feature, xshift=2cm, yshift=2.4cm] {Info. proc. 1};
\node (filter2) [block, below of=filter1, yshift=.4cm] {Info. proc. 2};
\node (filter3) [block, below of=filter2, yshift=.4cm] {Info. proc. 3};
\node (filter4) [block, below of=filter3, yshift=.4cm] {Info. proc. 4};
\node (filter5) [block, below of=filter4, yshift=.4cm] {Info. proc. 5};
\node (filter6) [block, below of=filter5, yshift=.4cm] {Info. proc. 6};
\node (filter7) [block, below of=filter6, yshift=.4cm] {Info. proc. 7};
\node (filter8) [block, below of=filter7, yshift=.4cm] {Info. proc. 8};
\node (filter9) [block, below of=filter8, yshift=.4cm] {Info. proc. 9};

\fill[blue!20!white, xshift=7cm, yshift=-1.5cm] (0, 0) rectangle (1, 1);
\fill[green!20!white, xshift=7cm, yshift=-1.5cm] (0, 1) rectangle (1, 2);
\fill[blue!20!white, xshift=7cm, yshift=-1.5cm] (0, 2) rectangle (1, 3);
\fill[green!20!white, xshift=7cm, yshift=-1.5cm] (1, 0) rectangle (2, 1);
\fill[blue!20!white, xshift=7cm, yshift=-1.5cm] (1, 1) rectangle (2, 2);
\fill[green!20!white, xshift=7cm, yshift=-1.5cm] (1, 2) rectangle (2, 3);
\fill[blue!20!white, xshift=7cm, yshift=-1.5cm] (2, 0) rectangle (3, 1);
\fill[green!20!white, xshift=7cm, yshift=-1.5cm] (2, 1) rectangle (3, 2);
\fill[blue!20!white, xshift=7cm, yshift=-1.5cm] (2, 2) rectangle (3, 3);

\draw[step=1cm, yshift=.5cm] (7,-2) grid (10, 1);

\draw[-to] (feature.east) -- (filter1.west);
\draw[-to] (feature.east) -- (filter2.west);
\draw[-to] (feature.east) -- (filter3.west);
\draw[-to] (feature.east) -- (filter4.west);
\draw[-to] (feature.east) -- (filter5.west);
\draw[-to] (feature.east) -- (filter6.west);
\draw[-to] (feature.east) -- (filter7.west);
\draw[-to] (feature.east) -- (filter8.west);
\draw[-to] (feature.east) -- (filter9.west);

\draw[-to] (filter1.east) -| (7.5, 1.2);
\draw[-to] (filter2.east) -| (8.5, 1.2);
\draw[-to] (filter3.east) -- (9.5, 1.2);
\draw[-to] (filter4.east) -| (7.5, 0.2);
\draw[-to] (filter5.east) -- (8.4, 0);
\draw[-to] (filter6.east) -| (9.5, -0.2);
\draw[-to] (filter7.east) -- (7.5, -1.2);
\draw[-to] (filter8.east) -| (8.5, -1.2);
\draw[-to] (filter9.east) -| (9.5, -1.2);

\node at (7.5, 1) {1};
\node at (8.5, 1) {2};
\node at (9.5, 1) {3};
\node at (7.5, 0) {4};
\node at (8.5, 0) {5};
\node at (9.5, 0) {6};
\node at (7.5, -1) {7};
\node at (8.5, -1) {8};
\node at (9.5, -1) {9};
\end{tikzpicture}

%% file: figs/partial1.tex
    \begin{tikzpicture}

\fill[blue!20!white] (0, 0) rectangle (3, 3);
\fill[green!20!white] (0, 3) rectangle (3, 6);
\fill[blue!20!white] (0, 6) rectangle (3, 9);
\fill[green!20!white] (3, 0) rectangle (6, 3);
\fill[blue!20!white] (3, 3) rectangle (6, 6);
\fill[green!20!white] (3, 6) rectangle (6, 9);
\fill[blue!20!white] (6, 0) rectangle (9, 3);
\fill[green!20!white] (6, 3) rectangle (9, 6);
\fill[blue!20!white] (6, 6) rectangle (9, 9);

\fill[blue!40!white] (1, 1) rectangle (2, 2);
\fill[blue!40!white] (4, 1) rectangle (5, 2);
\fill[blue!40!white] (7, 1) rectangle (8, 2);
\fill[blue!40!white] (1, 4) rectangle (2, 5);
\fill[blue!40!white] (7, 4) rectangle (8, 5);
\fill[blue!40!white] (1, 7) rectangle (2, 8);
\fill[blue!40!white] (4, 7) rectangle (5, 8);
\fill[blue!40!white] (7, 7) rectangle (8, 8);

\foreach \i in {4, 13, 22, 31}
\foreach \j in {4, 13, 22, 31}
{
\fill[blue] (\i, \j) rectangle (\i+1, \j+1);
}
\draw[step=1cm, gray!50!white] (0,0) grid (36, 36);
\end{tikzpicture}

%% file: figs/partial2.tex
    \begin{tikzpicture}

\fill[blue!20!white] (0, 0) rectangle (1, 1);
\fill[green!20!white] (0, 1) rectangle (1, 2);
\fill[blue!20!white] (0, 2) rectangle (1, 3);
\fill[green!20!white] (1, 0) rectangle (2, 1);
\fill[blue!20!white] (1, 1) rectangle (2, 2);
\fill[green!20!white] (1, 2) rectangle (2, 3);
\fill[blue!20!white] (2, 0) rectangle (3, 1);
\fill[green!20!white] (2, 1) rectangle (3, 2);
\fill[blue!20!white] (2, 2) rectangle (3, 3);

\foreach \i in {1, 4, 7, 10, 13, 16, 19, 22, 25, 28, 31, 34}
\foreach \j in {1, 4, 7, 10, 13, 16, 19, 22, 25, 28, 31, 34}
{
\fill[blue] (\i, \j) rectangle (\i+1, \j+1);
}

\foreach \i in {4, 13, 22, 31}
\foreach \j in {4, 13, 22, 31}
{
\fill[blue] (\i, \j) rectangle (\i+1, \j+1);
}
\draw[step=1cm, gray!50!white] (0,0) grid (36, 36);
\end{tikzpicture}

%% file: figs/layer_structure-alternative.tex
\tikzset{
	anchors1/.pic={
		\draw[step=1, lightgray] (0,0) grid (36, 36);
		\foreach \k in {0,...,3}
		\foreach \l in {0,...,3}
		{
			\fill[plot0] (9*\k+4, 9*\l+4) rectangle (9*\k+5, 9*\l+5);
		}
	},
	anchors2/.pic={
		\draw[step=1, lightgray] (0,0) grid (36, 36);
		\draw[step=9, gray] (0,0) grid (36, 36);
		\foreach \k in {0,...,3}
		\foreach \l in {0,...,3}
		{
		\foreach \i in {0,...,2}
		\foreach \j in {0,...,2}
		{
				\fill[plot1] (3*\i+1 + 9*\k, 3*\j+1 + 9*\l) rectangle (3*\i+2 + 9*\k, 3*\j + 2 + 9*\l);
		}
		\fill[plot0] (9*\k+4, 9*\l+4) rectangle (9*\k+5, 9*\l+5);
		}
	},
	anchors3/.pic={
		\draw[fill=plot1,fill opacity=.5] (0, 0) rectangle (36, 36);
		\draw[step=1, lightgray] (0,0) grid (36, 36);
		\draw[step=3, gray] (0,0) grid (36, 36);
		\foreach \k in {0,...,3}
		\foreach \l in {0,...,3}
		\foreach \i in {0,...,2}
		\foreach \j in {0,...,2}
		{
			\fill[plot0] (3*\i+1 + 9*\k, 3*\j+1 + 9*\l) rectangle (3*\i+2 + 9*\k, 3*\j + 2 + 9*\l);
		}
	}
}
\begin{tikzpicture}
\tikzstyle{layer} = [rectangle, rounded corners, minimum width=2.5cm, minimum height=.6cm, align=center, draw=black]

\node (input) [layer] {Feature $\boldsymbol{\Gamma}$};
\node (layer1) [layer,below=.5 of input] {Normal layer (16 RIS elements)};
\node (layer2) [layer, below=.1 of layer1] {Normal layer (16 RIS elements)};
\node (layer3) [layer, below=1 of layer2] {Expansion layer (144 RIS elements)};
\node (layer4) [layer, below=.1 of layer3] {Normal layer (144 RIS elements)};
\node (layer5) [layer, below=.1 of layer4] {Normal layer (144 RIS elements)};
\node (layer6) [layer, below=1 of layer5] {Expansion layer (1296 RIS elements)};
\node (layer7) [layer, below=.1 of layer6] {Normal layer (1296 RIS elements)};
\node (layer8) [layer, below=.1 of layer7] {Normal layer (1296 RIS elements)};
\node (output) [layer, below=.5 of layer8] {Solution $\boldsymbol{\Phi}$};

\pic[scale=.08] at ($(layer2)!.5!(layer3) + (3, 0)$) {anchors1};
\pic[scale=.08] at ($(layer5)!.5!(layer6) + (3, 0)$) {anchors2};
\pic[scale=.08] at ($(output.north) + (3, 0)$) {anchors3};

\draw[->] (input) -- (layer1);
\draw[->] (layer8) -- (output);
\draw[->] (layer2) -- (layer3);
\draw[->] (layer5) -- (layer6);
\end{tikzpicture}

%% file: joint.tex
\subsection{Joint Optimization of BS Precoding and RIS Configuration}
\label{sec:joint}

Problem~\eqref{eq:problem_sdma} involves joint optimization of \gls{bs} precoding and \gls{ris} configuration.
Using domain knowledge in communication,
we identify 
\gls{wmmse} precoding~\cite{shi2011iteratively} 
as a high-performance analytical precoding scheme.
We apply it
in order to guarantee precoding performance
and reduce training difficulty
since we do not need to optimize precoding.%

The \gls{wmmse} precoding~\cite{shi2011iteratively} is briefly elaborated as follows:
It is first proved that the \gls{wsr} maximization problem
is equivalent to a weighted sum \gls{mse} minimization problem.
Following this observation,
an iterative algorithm is proposed.
In each iteration,
the weight of each user's \gls{mse} is updated.
Subsequently,
the precoding matrix for each user is computed.
The iteration is terminated if the \gls{mse} weights converge.
In particular,
a factor~$\mu_k$ (equation~(18) in~\cite{shi2011iteratively})
is computed numerically in the precoding matrix computation
in order to meet the transmit power constraint~\eqref{eq:sdma_transmit_power}.
Therefore, the \gls{wmmse} precoder is indifferentiable
and cannot be part of the objective
because the \gls{nn} is trained with gradient ascent.
As a result, we apply \gls{ao},
where we fix the \gls{nn} and compute the precoding matrix~$\mathbf{V}$ using the \gls{ris} configuration generated by the current RISnet
for each data sample in $\mathcal{D}$,
then treat $\mathbf{V}$ as constants and train the \gls{nn} for the given precoding.

The framework of the hybrid approach of joint optimization is illustrated in \autoref{fig:hybrid}.

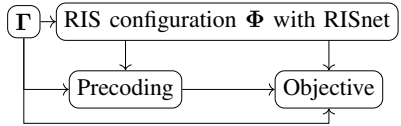
\begin{figure}[htbp]
    \centering
    \resizebox{.6\linewidth}{!}{\input{figs/framework}}
    \caption{Hybrid solution of analytical precoding and ML-enabled RIS configuration.}
    \label{fig:hybrid}
\end{figure}

The training objective is~\eqref{eq:sdma_objective}.
Note that although constraints~\eqref{eq:sdma_transmit_power}--\eqref{eq:sdma_offdiagonal} do not appear in the objective,
\eqref{eq:sdma_transmit_power} is fulfilled by the \gls{wmmse} algorithm
(see equation~(18) in~\cite{shi2011iteratively}),
\eqref{eq:sdma_passive_ris} is satisfied because ${|\e^{\imag\varphi}|=1}$ for any ${\varphi \in \mathbb{R}}$,
\eqref{eq:sdma_offdiagonal} is guaranteed because all the off-diagonal elements stay constant throughout training and inference.

%% file: figs/framework.tex
\begin{tikzpicture}
\tikzstyle{block} = [draw=black, rounded corners];

\node (risnet) [block] {RIS configuration $\boldsymbol{\Phi}$ with RISnet};
\node (precoding) [block, below of = risnet, xshift = -1.5cm] {Precoding};
\node (feature) [block, left of = risnet, xshift=-2cm] {$\boldsymbol{\Gamma}$};
\node (objective) [block, right of = precoding, xshift = 2cm] {Objective}; 

\draw[->] (feature) -- (risnet);
\draw[->] (feature) |- (precoding);
\draw[->] (precoding) -- (objective);
\draw[<-] (precoding) -- ++(0, .72);
\draw[<-] (objective) -- ++(0, .72);
\draw[->] (feature) -- ++(0, -1.5) -- ++(4.5, 0) -- (objective);
\end{tikzpicture}

%% file: algorithm.tex
Summarizing the above descriptions,
the algorithms to train the \gls{nn} is formulated as \autoref{alg:sdma}.

\begin{algorithm}
\caption{Neural network training with \gls{ao}}
\label{alg:sdma}
\begin{algorithmic}[1]
\State Initialize the permutation-invariant RISnet $N_\theta$.
\Repeat
\State Randomly choose data samples in a batch.
\State Compute \gls{wmmse} precoding matrix according to \cite{shi2011iteratively} for every data sample,
where the precoding matrix is considered as constants for training.
\State Compute phase shifts $\boldsymbol{\Phi}$ with current $N_\theta$.
\State Compute objective \eqref{eq:sdma_objective} with \gls{csi}, user weights, precoding (considered as constants) and phase shifts.
\State Compute gradient of \eqref{eq:sdma_objective} \gls{wrt} $\theta$ with backward propagation.
\State Perform an optimization with the ADAM optimizer according to the gradient.
\Until{\Gls{wsr} stops increasing.}
\end{algorithmic}
\end{algorithm}

%% file: results.tex
\section{Training and Testing Results}
\label{sec:results}

The training and testing results are presented in this section.
The open-source DeepMIMO data set~\cite{Alkhateeb2019} is applied to generate channel data.
The chosen urban scenario is shown in \autoref{fig:scenario},
where the \gls{los} channel from \gls{bs} to users are blocked by a building.
Only a weak direct channel is available through reflections on buildings and ground.
Furthermore, the channel from \gls{bs} to \gls{ris} has multiple \glspl{mpc} such that the \gls{mimo} channel matrix has a high rank to support multiple users in \gls{sdma}.
Finally,
we choose users at least 8m from each other to realize a full-rank channel matrix from \gls{ris} to users,
since the rank of the cascaded channel matrix is less than or equal to 
the lowest rank of the factors in the product,
see \autoref{re:rank}.
The user grouping is assumed given.
We note that it is an open topic to assign users according to channel/user positions.

\begin{figure}%
    \centering
    \resizebox{.35\linewidth}{!}{\input{figs/scenario}}
    \caption{The considered scenario: an intersection in an urban environment.}
    \label{fig:scenario}
\end{figure}
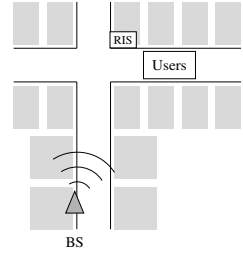

Important assumptions and parameter settings are listed in \autoref{tab:params}.
It is to note that the performance is sensitive to the learning rate.

\begin{table}[htbp]
    \centering
    \caption{Scenario and model parameters}
    \label{tab:params}
    \begin{tabularx}{\linewidth}{lX}
        \toprule
        Parameter & Value \\
        \midrule
        Number of \gls{bs} antennas & \num{9} \\
        \gls{ris} size & $36 \times 36$ elements\\
        Carrier frequency & $\SI{3.5}{\GHz}$\\
        Distance between adjacent antennas at \gls{bs} & 0.5 wavelength\\
        Distance between adjacent antennas at \gls{ris} & 0.25 wavelength\\
        Number of users & \num{4}\\
        Learning rate & $8 \times 10^{-4}$ -- $1.5\times 10^{-3}$\\
        Batch size & \num{512}\\
        Optimizer & ADAM\\
        Number of data samples in training set & \num{10240}\\
        Number of data samples in testing set & \num{1024}\\
        \bottomrule
    \end{tabularx}
\end{table}

As described in \autoref{sec:risnet_partial},
the effectiveness of \gls{ris} configuration with partial \gls{csi}
depends on the channel model:
the partial \gls{csi} contains sufficient information to configure the whole \gls{ris}
if the channel consist of a few \gls{los} or specular propagation paths.
In this case,
the channel gains at different \gls{ris} elements are
strongly correlated spatially
because all these channel gains are due to these specular propagation paths.
On the contrary,
if the channel has infinitely many and infinitely weak propagation paths due to scattering,
the channel gains at different \gls{ris} elements
are spatially \gls{iid}~\cite{jamali2022impact}.
We assume three channel models to assess the feasibility of applying partial \gls{csi} for \gls{ris} configuration:
\begin{itemize}
    \item Deterministic ray-tracing channel from DeepMIMO simulator, which is most feasible to infer the full \gls{csi} from the partial \gls{csi}.
    \item Deterministic ray-tracing channel plus \gls{iid} scattering gains on each \gls{ris} element.
    It is less feasible to infer the full \gls{csi} from the partial \gls{csi} with this model.
    \item \Gls{iid} channel model due to scattering of infinitely many infinitely weak propagation paths, 
    where the inference of full \gls{csi} from partial \gls{csi} is impossible.
\end{itemize}

\subsection{Training Behavior}

\autoref{fig:training_sdma_mutual_coupling} illustrates the improvement of \gls{wsr} in training and testing with full and partial \glspl{csi},
where the user weights are uniformly randomly generated and sum up to one,
and the data for testing is independently generated from the data for training.
It can be observed that training and testing with the same setup realize similar performances,
suggesting a good generalizability of the trained model.
From \autoref{fig:wsr_det},
we observe that similar performances are achieved with full and partial \gls{csi} when the channel is generated by the ray-tracing simulation
and is therefore deterministic,
suggesting that the partial \gls{csi} is sufficient and the difficulty of channel estimation can be relieved significantly.
With the deterministic channel and \gls{iid} scattering gain (\autoref{fig:wsr_semi}),
the realized \gls{wsr} with partial \gls{csi} is lower
than with full \gls{csi},
because the full \gls{csi} cannot be recovered from
the partial \gls{csi} due to the spatially \gls{iid} scattering gain.
However, the difference between full and partial \glspl{csi}
is insignificant,
suggesting that the proposed approach is robust against
\gls{iid} scattering gain.
If the channel is spatially fully \gls{iid}
(\autoref{fig:wsr_iid}),
the \gls{wsr} with full \gls{csi} is improved significantly,
whereas the \gls{wsr} with partial \gls{csi} stays almost constant.
This is because we cannot recover full \gls{csi} from partial \gls{csi} 
due to the independent channel gains at different \gls{ris} elements.

\begin{figure}[htbp]
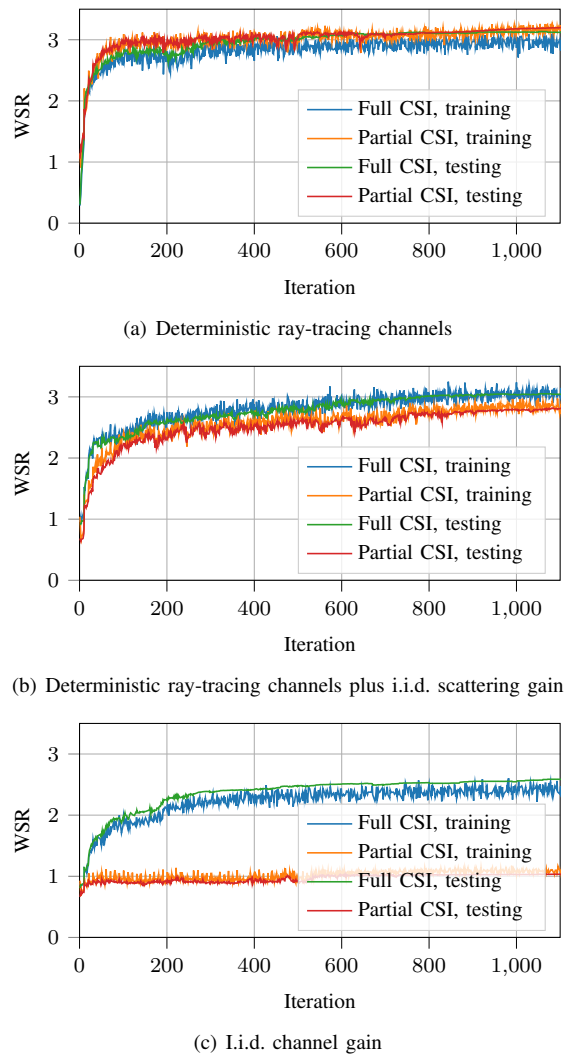

    \centering
    \subfigure[Deterministic ray-tracing channels\label{fig:wsr_det}]
    {\input{figs/sdma_0_mutual_coupling_rate}}
    \subfigure[Deterministic ray-tracing channels plus i.i.d. scattering gain\label{fig:wsr_semi}]
    {\input{figs/sdma_p_mutual_coupling_rate}}
    \subfigure[I.i.d. channel gain\label{fig:wsr_iid}]
    {\input{figs/sdma_iid_mutual_coupling_rate}}
    \caption{Realized \gls{wsr} with SDMA in training and testing.}
    \label{fig:training_sdma_mutual_coupling}
\end{figure}

According to~\cite{he2021wireless},
the wireless channel is \emph{sparse} in many typical scenarios,
i.e., the signal arrives the receiver via a few distinct \glspl{mpc}
and the spatially \gls{iid} scattering effect is very limited.
This fact is the foundation of many compressed sensing based channel estimation~\cite{haghighatshoar2017massive,wunder2019low}.
Since most real wireless channels are similar to deterministic channel model with or without \gls{iid} scattering gain 
(i.e., like \autoref{fig:wsr_det} and \autoref{fig:wsr_semi} rather than \autoref{fig:wsr_iid}),
the proposed method with partial \gls{csi} is believed to work well not only in simulation,
but also in reality.

\subsection{Comparison with Baselines}

Next, we compare our proposed approach with baselines.
Since \gls{ris} optimization considering mutual coupling for \gls{sdma} is still an open topic,
we assume an \gls{ris} without mutual coupling and use
\gls{drl}~\cite{huang2020reconfigurable},
random phase shift,
and \gls{bcd} algorithm~\cite{guo2020weighted}
as baselines for comparison.
The problem formulation is with full \gls{csi} and
without mutual coupling.
As a fair comparison with the same problem formulation,
we also train RISnet assuming no mutual coupling,
i.e., using~\eqref{eq:reduced_channel} as the channel model.
\autoref{fig:comparison_baselines} shows the performance comparison of the proposed approach and the baselines.
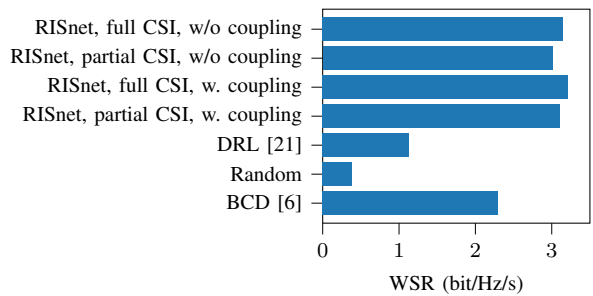
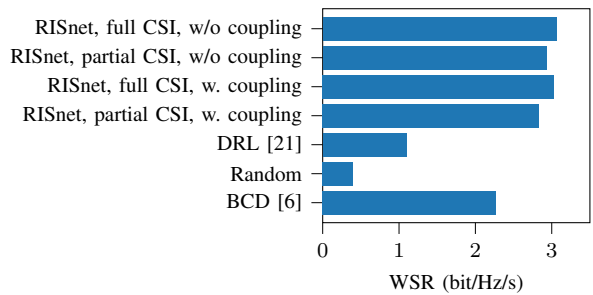
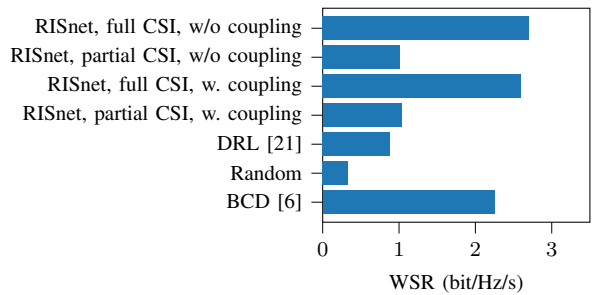
\begin{figure}[htbp]
    \centering
    \subfigure[Deterministic ray-tracing channels\label{fig:baselines_deteministic}]{
    \input{figs/comparison_deterministic_1}}
    \subfigure[Deterministic ray-tracing channels plus i.i.d. scattering gain\label{fig:baselines_semi}]{
    \input{figs/comparison_semi_1}}
    \subfigure[I.i.d. channel gain\label{fig:baselines_iid}]{
    \input{figs/comparison_iid_1}}
    \caption{Comparison between proposed approach with baselines.
    The proposed approach outperforms all baselines with the exception of partial \gls{csi} and \gls{iid} channel gain.}
    \label{fig:comparison_baselines}
\end{figure}
We have the following observations:
\begin{itemize}
    \item The proposed RISnet outrforms baseline algorithms \gls{drl}, 
    random phase shift
    and \gls{bcd} algorithm significantly.
    The only exception is RISnet with \gls{iid} channel gain and partial \gls{csi},
    as explained below.
    A main reason for the better performance of the proposed approach is
    that \gls{drl} with a fully connected \gls{nn}
    and \gls{bcd} algorithm have a 
    poor scalability with the number of \gls{ris} elements.
    \item The proposed RISnet with partial \gls{csi} works well with deterministic ray-tracing channel model (\autoref{fig:baselines_deteministic}),
    and deterministic channel model plus \gls{iid} channel gain
    (\autoref{fig:baselines_semi}).
    The observation holds for both setups with and without mutual coupling.
    This is because the partial \gls{csi} contains sufficient information to recover the full \gls{csi}.
    On the contrary,
    the full \gls{csi} cannot be recovered from the partial \gls{csi} with \gls{iid} channel gain.
    Therefore, the performance with partial \gls{csi} and \gls{iid} channel gain is poor.
\end{itemize}

\subsection{Necessity to Consider Mutual Coupling}

In this section,
we demonstrate the necessity to consider mutual coupling
by testing RISnet trained \emph{without mutual coupling} to problem \emph{with mutual coupling}. 
As shown in \autoref{fig:comparison_mutual_coupling},
the model mismatch
(i.e., model trained without mutual coupling and tested with mutual coupling)
results in a significant degradation particularly for RISnet with partial \gls{csi},
which justifies the necessity of explicit consideration of mutual coupling
if it exists in the \gls{ris} hardware.

\begin{figure}[htbp]
    \centering
    \subfigure[Deterministic ray-tracing channels\label{fig:baselines_deteministic_2}]{
    \input{figs/comparison_deterministic_2}}
    \subfigure[Deterministic ray-tracing channels plus i.i.d. scattering gain\label{fig:baselines_semi_2}]{
    \input{figs/comparison_semi_2}}
    \subfigure[I.i.d. channel gain\label{fig:baselines_iid_2}]{
    \input{figs/comparison_iid_2}}
    \caption{Comparison between testing results considering mutual coupling
    of trained models without (first two rows) and with (last two rows) consideration of mutual coupling.
    This result justifies the necessity to consider mutual coupling if it exists
    because the model mismatch results in significant performance loss.}
    \label{fig:comparison_mutual_coupling}
\end{figure}
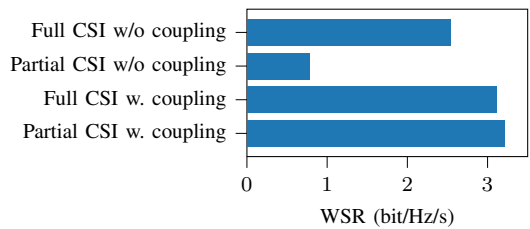
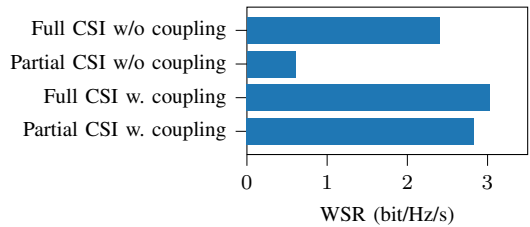
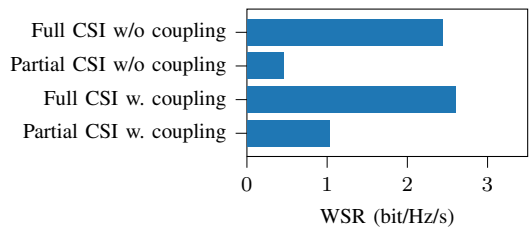

\subsection{Performance with Different Numbers of Anchor Elements}

In this section,
we investigate the impact of the anchor element number
on the performance.
We consider a \gls{ris} with 16 anchor elements with the ability of channel estimation,
as depicted in \autoref{fig:expansion},
and a \gls{ris} with 4 anchor elements.
We expect that the \gls{ris} with 4 anchor elements is more vulnerable to \gls{iid} scattering gain
due to the less input information.
The realized \gls{wsr} is
shown in \autoref{fig:comparison_anchor}.
The two \glspl{ris} are roughly equally good with deterministic channels
and roughly equally bad with \gls{iid} channel gain.
However,
the \gls{ris} with 16 anchor elements 
(more partial \gls{csi})
performs much better than the \gls{ris} with 4 anchor elements
(less partial \gls{csi}),
as we expect.
This result suggests that the required number of anchor elements depends on the channel property:
The more \gls{iid} scattering gain,
the more required anchor elements
and more expensive hardware.
The optimal choice of the hardware depends on the propagation environment.

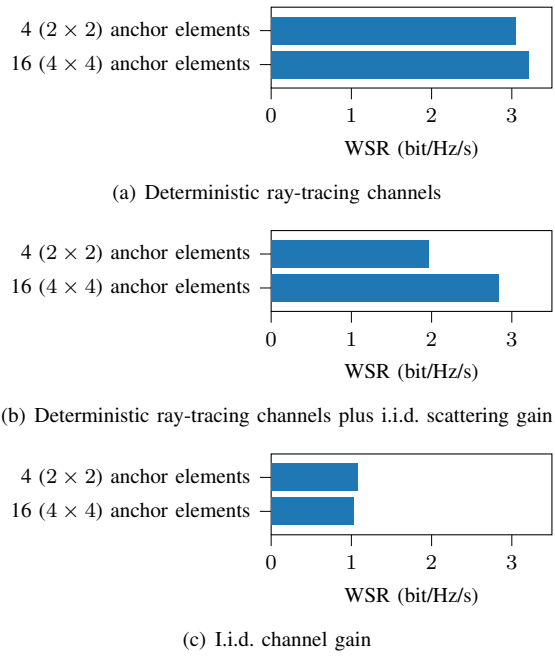
\begin{figure}[htbp]
    \centering
    \subfigure[Deterministic ray-tracing channels\label{fig:anchor_deteministic}]{
    \input{figs/comparison_deterministic_4}}
    \subfigure[Deterministic ray-tracing channels plus i.i.d. scattering gain\label{fig:anchor_semi}]{
    \input{figs/comparison_semi_4}}
    \subfigure[I.i.d. channel gain\label{fig:anchor_iid}]{
    \input{figs/comparison_iid_4}}
    \caption{Comparison between training results with different numbers of anchor elements.
    More anchor elements have higher hardware complexity but are more robust to \gls{iid} channel gain.}
    \label{fig:comparison_anchor}
\end{figure}

\subsection{Complexity Analysis}
In addition to performance,
another major advantage of RISnet is the low complexity in real-time application.
According to~\eqref{eq:tensor_operation2} and~\eqref{eq:tensor_operation4},
the information processing of a layer is considered as cascading operation of 
tensor multiplication and addition of complexity~${\mathcal{O}(QPKN) + \mathcal{O}(QKN)}$,
where $P$ is the input feature dimension and 
$Q$ is the output feature dimension,
the ReLU activation of complexity~$\mathcal{O}(QKN)$
(because the activation function is performed element-wise),
and the averaging operation of complexity~$\mathcal{O}(N)$, $\mathcal{O}(U)$ and $\mathcal{O}(NU)$
for class~\texttt{ca}, \texttt{oc} and \texttt{oa},
respectively.
Asymptotically,
the complexity of RISnet inference is $\mathcal{O}(QPKN)$.
On the other hand,
the asymptotical complexity of the \gls{bcd} algorithm is
$\mathcal{O}(U^2N^2)$~\cite{guo2020weighted}.
The asymptotical complexity of a fully connected \gls{nn} applied in~\cite{huang2020reconfigurable} is $\mathcal{O}(UN^2)$
since the input dimension is proportional to $UN$ 
(number of users times number of \gls{ris} elements)
and the output dimension is $N$.
We can see clearly that the complexity of RISnet inference grows
linearly with the number of \gls{ris} elements~$N$,
while the complexities of baseline approaches~\cite{guo2020weighted,huang2020reconfigurable} grow quadratically with $N$.
This observation confirms the high scalability of the proposed RISnet.

Practically,
RISnet can configure a \gls{ris} with 1296~elements within 0.07~second on an off-the-shelf laptop with Apple~M3~Pro processor,
while reference~\cite{guo2020weighted} needs 253~seconds
and reference~\cite{huang2020reconfigurable} needs 0.23-second.
The reason of the fast inference is that the most computation complexity is in training,
which takes about 8~hours.
Once the training is finished,
the inference (application) with a trained RISnet is very efficient.
Compared to a fully connected \gls{nn},
where the number of trainable parameters is proportional to the product
of input and output dimensions,
the number of trainable parameters of RISnet is independent from the input and output dimensions.
This comparison shows the advantage of complexity and real-time performance of the proposed RISnet.
In future applications,
we can leverage the hardware developed for heavy \gls{ml} applications such as \glspl{llm} and autonomous driving
in a much smaller scale,
because the proposed RISnet has about $10^4$ trainable parameters,
which is one million times less than current \glspl{llm}~\cite{chang2024survey},
making an efficient real-time application feasible with low-cost hardware.

%% file: figs/scenario.tex
		\begin{tikzpicture}[scale=1]
            \tikzstyle{base}=[isosceles triangle, draw, rotate=90, fill=gray!60, minimum size =0.12cm]
   
			\foreach \i in {1, 1.8, 3.4, 4.2, 5.0, 5.8}
			\foreach \j in {2, 4}
			{
				\fill[gray!30!white] (\i, \j) rectangle (\i+0.6, \j+1);
			}
   
			\foreach \i in {1.4, 3.4}
			\foreach \j in {-0.4, 0.8}
			{
				\fill[gray!30!white] (\i, \j) rectangle (\i+1, \j+1);
			}
			\node[align=center, rectangle,draw, minimum height=.65cm, text width=1cm] (ue) at (4.7,3.5) {Users};
			\draw (1,3.9) -- (2.5,3.9);
			\draw (1,3.1) -- (2.5,3.1);
			\draw (3.3,3.9) -- (6.4,3.9);
			\draw (3.3,3.1) -- (6.4,3.1);
			
			\draw (2.5,-0.4) -- (2.5,3.1);
			\draw (3.3,-0.4) -- (3.3,3.1);
			\draw (2.5,3.9) -- (2.5,5);
			\draw (3.3,3.9) -- (3.3,5);
			
            \node[base] (BS) at (2.45,0.1){};
			\node[below of=BS, yshift=.2cm] (bs) {BS};
            \draw[decoration=expanding waves,decorate] (BS) -- (3.2,1.9, 1.2);
			\node[rectangle, fill=white, draw, scale=.8] (ris) at (3.6,4.1)
			{  \hspace{.0cm}RIS\hspace{.0cm} };
	\end{tikzpicture}

%% file: figs/comparison_deterministic_1.tex
\begin{tikzpicture}

\definecolor{darkgray176}{RGB}{176,176,176}
\definecolor{steelblue31119180}{RGB}{31,119,180}
\tikzstyle{every node}=[font=\footnotesize]

\begin{axis}[
height=.5\linewidth,
tick align=outside,
tick pos=left,
x grid style={darkgray176},
xlabel={WSR (bit/Hz/s)},
xmin=0, xmax=3.5,
xtick style={color=black},
y grid style={darkgray176},
ymin=-0.69, ymax=6.69,
ytick style={color=black},
ytick={0,1,2,3,4,5, 6, 7},
yticklabels={
  BCD~\cite{guo2020weighted},
  Random,
  DRL~\cite{huang2020reconfigurable},
  {RISnet, partial CSI, w. coupling},
  {RISnet, full CSI, w. coupling},
  {RISnet, partial CSI, w/o coupling},
  {RISnet, full CSI, w/o coupling},
}
]
\draw[draw=none,fill=steelblue31119180] (axis cs:0,-0.4) rectangle (axis cs:2.30,0.4);  %
\draw[draw=none,fill=steelblue31119180] (axis cs:0,0.6) rectangle (axis cs:0.382,1.4);  %
\draw[draw=none,fill=steelblue31119180] (axis cs:0,1.6) rectangle (axis cs:1.13,2.4);  %
\draw[draw=none,fill=steelblue31119180] (axis cs:0,2.6) rectangle (axis cs:3.11,3.4);  %
\draw[draw=none,fill=steelblue31119180] (axis cs:0,3.6) rectangle (axis cs:3.21,4.4);  %
\draw[draw=none,fill=steelblue31119180] (axis cs:0,4.6) rectangle (axis cs:3.01,5.4);  %
\draw[draw=none,fill=steelblue31119180] (axis cs:0,5.6) rectangle (axis cs:3.15,6.4);  %
\end{axis}

\end{tikzpicture}

%% file: figs/comparison_semi_1.tex
\begin{tikzpicture}

\definecolor{darkgray176}{RGB}{176,176,176}
\definecolor{steelblue31119180}{RGB}{31,119,180}
\tikzstyle{every node}=[font=\footnotesize]

\begin{axis}[
height=.5\linewidth,
tick align=outside,
tick pos=left,
x grid style={darkgray176},
xlabel={WSR (bit/Hz/s)},
xmin=0, xmax=3.5,
xtick style={color=black},
y grid style={darkgray176},
ymin=-0.69, ymax=6.69,
ytick style={color=black},
ytick={0,1,2,3,4,5, 6, 7},
yticklabels={
  BCD~\cite{guo2020weighted},
  Random,
  DRL~\cite{huang2020reconfigurable},
  {RISnet, partial CSI, w. coupling},
  {RISnet, full CSI, w. coupling},
  {RISnet, partial CSI, w/o coupling},
  {RISnet, full CSI, w/o coupling},
}
]
\draw[draw=none,fill=steelblue31119180] (axis cs:0,-0.4) rectangle (axis cs:2.27,0.4);  %
\draw[draw=none,fill=steelblue31119180] (axis cs:0,0.6) rectangle (axis cs:0.40,1.4);  %
\draw[draw=none,fill=steelblue31119180] (axis cs:0,1.6) rectangle (axis cs:1.10,2.4);  %
\draw[draw=none,fill=steelblue31119180] (axis cs:0,2.6) rectangle (axis cs:2.83,3.4);  %
\draw[draw=none,fill=steelblue31119180] (axis cs:0,3.6) rectangle (axis cs:3.03,4.4);  %
\draw[draw=none,fill=steelblue31119180] (axis cs:0,4.6) rectangle (axis cs:2.93,5.4);  %
\draw[draw=none,fill=steelblue31119180] (axis cs:0,5.6) rectangle (axis cs:3.06,6.4);  %
\end{axis}

\end{tikzpicture}

%% file: figs/comparison_iid_1.tex
\begin{tikzpicture}

\definecolor{darkgray176}{RGB}{176,176,176}
\definecolor{steelblue31119180}{RGB}{31,119,180}
\tikzstyle{every node}=[font=\footnotesize]

\begin{axis}[
height=.5\linewidth,
tick align=outside,
tick pos=left,
x grid style={darkgray176},
xlabel={WSR (bit/Hz/s)},
xmin=0, xmax=3.5,
xtick style={color=black},
y grid style={darkgray176},
ymin=-0.69, ymax=6.69,
ytick style={color=black},
ytick={0,1,2,3,4,5, 6, 7},
yticklabels={
  BCD~\cite{guo2020weighted},
  Random,
  DRL~\cite{huang2020reconfigurable},
  {RISnet, partial CSI, w. coupling},
  {RISnet, full CSI, w. coupling},
  {RISnet, partial CSI, w/o coupling},
  {RISnet, full CSI, w/o coupling},
}
]
\draw[draw=none,fill=steelblue31119180] (axis cs:0,-0.4) rectangle (axis cs:2.25,0.4);  %
\draw[draw=none,fill=steelblue31119180] (axis cs:0,0.6) rectangle (axis cs:0.33,1.4);  %
\draw[draw=none,fill=steelblue31119180] (axis cs:0,1.6) rectangle (axis cs:0.88,2.4);  %
\draw[draw=none,fill=steelblue31119180] (axis cs:0,2.6) rectangle (axis cs:1.03,3.4);  %
\draw[draw=none,fill=steelblue31119180] (axis cs:0,3.6) rectangle (axis cs:2.6,4.4);  %
\draw[draw=none,fill=steelblue31119180] (axis cs:0,4.6) rectangle (axis cs:1.01,5.4);  %
\draw[draw=none,fill=steelblue31119180] (axis cs:0,5.6) rectangle (axis cs:2.70,6.4);  %
\end{axis}

\end{tikzpicture}

%% file: figs/comparison_deterministic_2.tex
\begin{tikzpicture}

\definecolor{darkgray176}{RGB}{176,176,176}
\definecolor{steelblue31119180}{RGB}{31,119,180}
\tikzstyle{every node}=[font=\footnotesize]

\begin{axis}[
height=.4\linewidth,
width=.6\linewidth,
tick align=outside,
tick pos=left,
x grid style={darkgray176},
xlabel={WSR (bit/Hz/s)},
xmin=0, xmax=3.5,
xtick style={color=black},
y grid style={darkgray176},
ymin=-0.69, ymax=3.69,
ytick style={color=black},
ytick={0,1,2,3},
yticklabels={
  {Partial CSI w. coupling},
  {Full CSI w. coupling},
  {Partial CSI w/o coupling},
  {Full CSI w/o coupling},
}
]
\draw[draw=none,fill=steelblue31119180] (axis cs:0,-0.4) rectangle (axis cs:3.21,0.4);  %
\draw[draw=none,fill=steelblue31119180] (axis cs:0,0.6) rectangle (axis cs:3.11,1.4);  %
\draw[draw=none,fill=steelblue31119180] (axis cs:0,1.6) rectangle (axis cs:0.78,2.4);  %
\draw[draw=none,fill=steelblue31119180] (axis cs:0,2.6) rectangle (axis cs:2.54,3.4);  %
\end{axis}

\end{tikzpicture}

%% file: figs/comparison_semi_2.tex
\begin{tikzpicture}

\definecolor{darkgray176}{RGB}{176,176,176}
\definecolor{steelblue31119180}{RGB}{31,119,180}
\tikzstyle{every node}=[font=\footnotesize]

\begin{axis}[
height=.4\linewidth,
width=.6\linewidth,
tick align=outside,
tick pos=left,
x grid style={darkgray176},
xlabel={WSR (bit/Hz/s)},
xmin=0, xmax=3.5,
xtick style={color=black},
y grid style={darkgray176},
ymin=-0.69, ymax=3.69,
ytick style={color=black},
ytick={0,1,2,3},
yticklabels={
  {Partial CSI w. coupling},
  {Full CSI w. coupling},
  {Partial CSI w/o coupling},
  {Full CSI w/o coupling},
}
]
\draw[draw=none,fill=steelblue31119180] (axis cs:0,-0.4) rectangle (axis cs:2.83,0.4);  %
\draw[draw=none,fill=steelblue31119180] (axis cs:0,0.6) rectangle (axis cs:3.03,1.4);  %
\draw[draw=none,fill=steelblue31119180] (axis cs:0,1.6) rectangle (axis cs:0.61,2.4);  %
\draw[draw=none,fill=steelblue31119180] (axis cs:0,2.6) rectangle (axis cs:2.4,3.4);  %
\end{axis}

\end{tikzpicture}

%% file: figs/comparison_iid_2.tex
\begin{tikzpicture}

\definecolor{darkgray176}{RGB}{176,176,176}
\definecolor{steelblue31119180}{RGB}{31,119,180}
\tikzstyle{every node}=[font=\footnotesize]

\begin{axis}[
height=.4\linewidth,
width=.6\linewidth,
tick align=outside,
tick pos=left,
x grid style={darkgray176},
xlabel={WSR (bit/Hz/s)},
xmin=0, xmax=3.5,
xtick style={color=black},
y grid style={darkgray176},
ymin=-0.69, ymax=3.69,
ytick style={color=black},
ytick={0,1,2,3},
yticklabels={
  {Partial CSI w. coupling},
  {Full CSI w. coupling},
  {Partial CSI w/o coupling},
  {Full CSI w/o coupling},
}
]
\draw[draw=none,fill=steelblue31119180] (axis cs:0,-0.4) rectangle (axis cs:1.03,0.4);  %
\draw[draw=none,fill=steelblue31119180] (axis cs:0,0.6) rectangle (axis cs:2.6,1.4);  %
\draw[draw=none,fill=steelblue31119180] (axis cs:0,1.6) rectangle (axis cs:0.46,2.4);  %
\draw[draw=none,fill=steelblue31119180] (axis cs:0,2.6) rectangle (axis cs:2.44,3.4);  %
\end{axis}

\end{tikzpicture}

%% file: figs/comparison_deterministic_4.tex
\begin{tikzpicture}

\definecolor{darkgray176}{RGB}{176,176,176}
\definecolor{steelblue31119180}{RGB}{31,119,180}
\tikzstyle{every node}=[font=\footnotesize]

\begin{axis}[
height=.3\linewidth,
width=.6\linewidth,
tick align=outside,
tick pos=left,
x grid style={darkgray176},
xlabel={WSR (bit/Hz/s)},
xmin=0, xmax=3.5,
xtick style={color=black},
y grid style={darkgray176},
ymin=-0.69, ymax=1.69,
ytick style={color=black},
ytick={0,1,2,3},
yticklabels={
  {16 ($4\times 4$) anchor elements},
  {4 ($2\times 2$) anchor elements},
}
]
\draw[draw=none,fill=steelblue31119180] (axis cs:0,-0.4) rectangle (axis cs:3.21,0.4);  %
\draw[draw=none,fill=steelblue31119180] (axis cs:0,0.6) rectangle (axis cs:3.05, 1.4);  %
\end{axis}

\end{tikzpicture}

%% file: figs/comparison_semi_4.tex
\begin{tikzpicture}

\definecolor{darkgray176}{RGB}{176,176,176}
\definecolor{steelblue31119180}{RGB}{31,119,180}
\tikzstyle{every node}=[font=\footnotesize]

\begin{axis}[
height=.3\linewidth,
width=.6\linewidth,
tick align=outside,
tick pos=left,
x grid style={darkgray176},
xlabel={WSR (bit/Hz/s)},
xmin=0, xmax=3.5,
xtick style={color=black},
y grid style={darkgray176},
ymin=-0.69, ymax=1.69,
ytick style={color=black},
ytick={0,1,2,3},
yticklabels={
  {16 ($4\times 4$) anchor elements},
  {4 ($2\times 2$) anchor elements},
}
]
\draw[draw=none,fill=steelblue31119180] (axis cs:0,-0.4) rectangle (axis cs:2.83,0.4);  %
\draw[draw=none,fill=steelblue31119180] (axis cs:0,0.6) rectangle (axis cs:1.96, 1.4);  %
\end{axis}

\end{tikzpicture}

%% file: figs/comparison_iid_4.tex
\begin{tikzpicture}

\definecolor{darkgray176}{RGB}{176,176,176}
\definecolor{steelblue31119180}{RGB}{31,119,180}
\tikzstyle{every node}=[font=\footnotesize]

\begin{axis}[
height=.3\linewidth,
width=.6\linewidth,
tick align=outside,
tick pos=left,
x grid style={darkgray176},
xlabel={WSR (bit/Hz/s)},
xmin=0, xmax=3.5,
xtick style={color=black},
y grid style={darkgray176},
ymin=-0.69, ymax=1.69,
ytick style={color=black},
ytick={0,1,2,3},
yticklabels={
  {16 ($4\times 4$) anchor elements},
  {4 ($2\times 2$) anchor elements},
}
]
\draw[draw=none,fill=steelblue31119180] (axis cs:0,-0.4) rectangle (axis cs:1.03,0.4);  %
\draw[draw=none,fill=steelblue31119180] (axis cs:0,0.6) rectangle (axis cs:1.08, 1.4);  %
\end{axis}

\end{tikzpicture}

%% file: conclusion.tex
\section{Conclusion}
\label{sec:conclusion}

The \gls{sdma} technique is crucial in multi-user wireless communication system.
Its performance depends strongly on the channel property,
which can be improved by the \gls{ris}.
In the previous research,
scalability of \gls{ris} elements,
unrealistic assumption of full \gls{csi}
and ignorance of mutual coupling between adjacent \gls{ris} elements
are the main limitations of realizing \gls{ris} in reality.
In this work, we have focused on deriving scalable solutions with unsupervised \gls{ml} for \gls{ris} configuration while making realistic assumptions regarding \gls{csi} knowledge and mutual coupling between adjacent \gls{ris} elements. 
We integrate domain-knowledge in communication and \gls{ml} techniques to design a problem-specific \gls{nn} architecture \emph{RISnet},
which is permutation-invariant to user input.
We further showed that partial \gls{csi} is sufficient to achieve a similar performance to full \gls{csi} if the channel comprises of a few specular propagation paths 
(i.e., the channel gain is not dominated by \gls{iid} components due to scattering).
Finally, we demonstrated that the proposed approach outperforms the baselines significantly and
it is necessary to explicitly consider the mutual coupling.
Beyond this work,
problem-specific \gls{ml} combining domain knowledge and \gls{ml} techniques provides unique opportunity to improve performance, complexity and equivariance compared to generic \gls{ml}.
This work can be extended by considering \gls{rsma} and beyond-diagonal \gls{ris}.

Code and data in this paper are available under \url{https://github.com/bilepeng/risnet_mutual_partial}.

%% file: appendix.tex
\section{Proof of \autoref{theorem:permutation_invariance}}
\label{appendix}

We first prove that every layer except the last layer is permutation-equivariant,
i.e., if
$\mathbf{F}_i'=\mathbf{P}\cdot\mathbf{F}_i$,
where $\mathbf{P}$ is an arbitrary permutation matrix and the multiplication is between $\mathbf{P}$ and the last two dimensions of $\mathbf{F}_i$,
$\mathbf{F}_{i+1}=M(\mathbf{F}_i)$ and 
$\mathbf{F}_{i+1}'=M(\mathbf{F}_i')$,
where $M$ is a layer,
we have $\mathbf{F}_{i+1}' = \mathbf{P}\cdot\mathbf{F}_{i+1}$.

We first consider class \texttt{cc},
the output of information processing unit \texttt{cc} of layer~$i$ given the permuted input is %
\begin{align}
    \mathbf{F}'^{\texttt{cc}}_{i+1} &= \text{ReLU}(\mathbf{W}^{\texttt{cc}}_{i} \mathbf{P} \cdot \mathbf{F}_{i} + \mathbf{b}_i^{\texttt{cc}})\\
    &= \text{ReLU}(\mathbf{P} \cdot \mathbf{W}^{\texttt{cc}}_{i} \mathbf{F}_{i} + \mathbf{b}_i^{\texttt{cc}})\\
    &= \mathbf{P} \cdot\text{ReLU}( \mathbf{W}^{\texttt{cc}}_{i} \mathbf{F}_{i} + \mathbf{b}_i^{\texttt{cc}})\\
    &= \mathbf{P} \mathbf{F}_{i+1}^{\text{\texttt{cc}}}.
    \label{eq:equivariance_cc}
\end{align}
The second line holds because the multiplication with $\mathbf{W}_i^{\texttt{cc}}$ is in the first dimension of $\mathbf{P}\cdot \mathbf{F}_i$,
whereas the multiplication with $\mathbf{P}$ is in the second and the third dimensions of $\mathbf{F}_i$.
The third line holds because ReLU is an elementwise operation.
The fourth line is the definition of $\mathbf{F}_{i+1}^{\text{\texttt{cc}}}$.
Similarly, we can prove $\mathbf{F}_{i+1}'^{\text{\texttt{ca}}} = \mathbf{P} \cdot \mathbf{F}_{i+1}^{\text{\texttt{ca}}}$.

To prove $\mathbf{F}_{i+1}'^{\text{\texttt{oc}}} = \mathbf{P} \cdot \mathbf{F}_{i+1}^{\text{\texttt{oc}}}$,
we first need to prove the multiplication between $\mathbf{E}^{U\times U}$ and $\mathbf{P}$ is commutative:
\begin{align}
        \mathbf{E}^{U\times U}\mathbf{P}
        &= \mathbf{1}^{U\times U}\mathbf{P} - \mathbf{I}^{U\times U} \mathbf{P}\\
        &= \mathbf{P}\mathbf{1}^{U\times U} - \mathbf{P}\mathbf{I}^{U\times U} \\
        &= \mathbf{P} \mathbf{E}^{U\times U}.
    \label{eq:commutative}
\end{align}
The third line holds because the sum of every row/column of $\mathbf{P}$ is 1.
We can now prove the permutation-equivariance as
\begin{align*}
    \mathbf{F}'^{\texttt{oc}}_{i+1} &= \mathbf{E}^{U \times U} \cdot \text{ReLU}(\mathbf{W}^{\texttt{oc}}_{i} 
    \mathbf{P}\cdot
    \mathbf{F}_{i} + \mathbf{b}_i^{\texttt{oc}}) \big/ (U-1)\\
    &= \mathbf{E}^{U \times U} \cdot \text{ReLU}(
    \mathbf{P}\cdot
    \mathbf{W}^{\texttt{oc}}_{i} 
    \mathbf{F}_{i} + \mathbf{b}_i^{\texttt{oc}}) \big/ (U-1)\\
    &= \mathbf{E}^{U \times U}
    \mathbf{P}\cdot
    \text{ReLU}(
    \mathbf{W}^{\texttt{oc}}_{i} 
    \mathbf{F}_{i} + \mathbf{b}_i^{\texttt{oc}}) \big/ (U-1)\\
    &= \mathbf{P}\cdot\mathbf{E}^{U \times U}
    \text{ReLU}(
    \mathbf{W}^{\texttt{oc}}_{i} 
    \mathbf{F}_{i} + \mathbf{b}_i^{\texttt{oc}}) \big/ (U-1)\\
    &= \mathbf{P}\cdot \mathbf{F}_{i+1}^{\text{\texttt{oc}}}.
    \label{eq:equivariance_oc}
\end{align*}
Similarly, we can prove $\mathbf{F}_{i+1}'^{\text{\texttt{oa}}} = \mathbf{P} \cdot \mathbf{F}_{i+1}^{\text{\texttt{oa}}}$.

Since the output of every class in every layer is permutation-equivariant,
the whole output after summation over users is permutation-invariant.